\documentclass[fleqn,11pt,twoside]{article}

\usepackage{amsthm,amsthm,amssymb, color, xcolor,epsfig, graphics}

\usepackage{amsmath, graphicx, latexsym, lscape}

\makeatletter
\newcommand{\copyrightnote}[2]{{\renewcommand{\thefootnote}{}
 \footnotetext{\small\it
\begin{flushleft}
 \copyright \ #1   #2  
\end{flushleft}}}}

\newcommand{\Name}[1]{\begin{flushleft}
                       \LARGE \bf #1
                       \end{flushleft}\vspace{-3mm}}

\newcommand{\Author}[1]{\begin{flushleft}
                       \it #1 \end{flushleft}}

\newcommand{\Address}[1]{\begin{flushleft}
                       \it #1 \end{flushleft}}

\newcommand{\Date}[1]{\begin{flushleft}
                      \small  \it #1 \end{flushleft}}

\newcommand{\evenhead}{Author \ name}
\newcommand{\oddhead}{Article \ name}

\renewcommand{\@evenhead}{
\hspace*{-3pt}\raisebox{-15pt}[\headheight][0pt]{\vbox{\hbox to \textwidth
{\thepage \hfil \evenhead}\vskip4pt \hrule}}}
\renewcommand{\@oddhead}{
\hspace*{-3pt}\raisebox{-15pt}[\headheight][0pt]{\vbox{\hbox to \textwidth
{\oddhead \hfil \thepage}\vskip4pt\hrule}}}
\renewcommand{\@evenfoot}{}
\renewcommand{\@oddfoot}{}

\long\def\@makecaption#1#2{%
  \vskip\abovecaptionskip
  \sbox\@tempboxa{\small \textbf{#1.}\ \ #2}%
  \ifdim \wd\@tempboxa >\hsize
    {\small \textbf{#1.}\ \ #2}\par
  \else
    \global \@minipagefalse
    \hb@xt@\hsize{\hfil\box\@tempboxa\hfil}%
  \fi
  \vskip\belowcaptionskip}

\newcommand{\JNMPnumberwithin}[3][\arabic]{%
  \@ifundefined{c@#2}{\@nocounterr{#2}}{%
    \@ifundefined{c@#3}{\@nocnterr{#3}}{%
      \@addtoreset{#2}{#3}%
      \@xp\xdef\csname the#2\endcsname{%
        \@xp\@nx\csname the#3\endcsname .\@nx#1{#2}}}}%
}

\renewenvironment{proof}[1][\proofname]{\par
  \normalfont
  \topsep6\p@\@plus6\p@ \trivlist
  \item[\hskip\labelsep\textbf{%
    #1\@addpunct{.}}]\ignorespaces
}{%
  \qed\endtrivlist
}

\newcommand{\resetfootnoterule} {
  \renewcommand\footnoterule{%
  \kern-3\p@
  \hrule\@width.4\columnwidth
  \kern2.6\p@}
}

\renewcommand{\footnoterule}{}

\makeatother

\theoremstyle{definition}

\newtheorem{thm}{Theorem}
\newtheorem{prop}[thm]{Proposition}
\newtheorem{cor}[thm]{Corollary}
\newtheorem{lem}[thm]{Lemma}
\theoremstyle{definition}
\newtheorem{defn}[thm]{Definition}

\newcommand{\I}{\rm{Im}}

\newcommand{\Span}{\rm{Span}}

\newcommand{\res}{{\rm res}}

\newcommand\fH{{\mathfrak{H}}}
\newcommand\fG{{\mathfrak{G}}}
\newcommand\fh{{\mathfrak{H}}}
\newcommand\fF{{\mathfrak{F}}}
\newcommand\fB{{\mathfrak{B}}}
\newcommand\fD{{\mathfrak{D}}}
\newcommand\fA{{\mathfrak{A}}}
\newcommand\fC{{\mathfrak{C}}}
\newcommand\fI{{\mathfrak{I}}}
\newcommand\fQ{{\mathfrak{Q}}}
\newcommand\fJ{{\mathfrak{Q}}}

\newcommand\cG{{\mathcal G}}

\newcommand\cH{{\mathcal H}}

\newcommand\cQ{{\mathcal Q}}
\newcommand\cA{{\mathcal A}}

\newcommand\cD{{\mathcal D}}
\newcommand\cL{{\mathcal L}}
\newcommand\cT{{\mathcal T}}

\def\bbbC{{\mathbb{C}}}
\def\bbbN{{\mathbb{N}}}
\def\bbbZ{{\mathbb{Z}}}

\begin{document}

\renewcommand{\evenhead}{ {\LARGE\textcolor{blue!10!black!40!green}{{\sf \ \ \ ]ocnmp[}}}\strut\hfill V M Buchstaber and A V Mikhailov}
\renewcommand{\oddhead}{ {\LARGE\textcolor{blue!10!black!40!green}{{\sf ]ocnmp[}}}\ \ \ \ \  KdV hierarchies and quantum Novikov equations}

\thispagestyle{empty}
\newcommand{\FistPageHead}[3]{
\begin{flushleft}
\raisebox{8mm}[0pt][0pt]
{\footnotesize \sf
\parbox{150mm}{{Open Communications in Nonlinear Mathematical Physics}\ \ \ \ {\LARGE\textcolor{blue!10!black!40!green}{]ocnmp[}}
\quad Special Issue 1, 2024\ \  pp
#2\hfill {\sc #3}}}\vspace{-13mm}
\end{flushleft}}

\FistPageHead{1}{\pageref{firstpage}--\pageref{lastpage}}{ \ \ }

\strut\hfill

\strut\hfill

\copyrightnote{The author(s). Distributed under a Creative Commons Attribution 4.0 International License}

\begin{center}
{  {\bf This article is part of an OCNMP Special Issue\\ 
\smallskip
in Memory of Professor Decio Levi}}
\end{center}

\smallskip

\Name{KdV hierarchies and quantum Novikov's equations.}

\Author{Victor~M.~Buchstaber$^{\,1}$ and Alexander~V.~Mikhailov$^{\,2}$}

\Address{$^{1}$ Steklov Mathematical Institute, RAS, Moscow, Russia\\[2mm]
$^{2}$ University of Leeds, UK}

\Date{Received August 11, 2023; Accepted December 13, 2023}

\setcounter{equation}{0}

\begin{abstract}

\noindent 
The paper begins with a review of the well known KdV hierarchy, $N$-th Novikov 
equations and its finite hierarchiey in the classical commutative case.
Its finite hierarchy consists of $N$ compatible integrable polynomial dynamical 
systems in $\bbbC^{2N}$.
Then we discuss a non-commutative version of the   $N$-th Novikov 
hierarchy defined on a finitely generated
free associative algebra $\fB_N$ with  $2N$ generators. Using the quantisation 
ideals method in $\fB_N$, 
for $N=1,2,3,4$, we have found two-sided
homogeneous ideals $\fJ_N\subset\fB_N$ (quantisation ideals) which are invariant 
with respect
to the $N$-th Novikov equation and such that the quotient algebra $\fC_N = 
\fB_N/ \fJ_N$ has a well defined
Poincare--Birkhoff--Witt basis. It enables us to define the quantum $N$-th 
Novikov equation and its hierarchy on the $\fC_N$. We have found $N$ 
commuting quantum first integrals (Hamiltonians) and represented equations of 
the hierarchy in the Heisenberg form. In this paper we introduce the concept of
Frobenius-Hochschild algebras and in its terms we express explicitly first 
integrals of the   $N$-th Novikov hierarchy in the commutative, free and 
quantum cases.

\end{abstract}

\label{firstpage}


\section{Introduction}

We dedicate this paper to the late Decio Levi, in tribute to his profound 
contributions to the advancement of the theory of integrable systems. His 
research journey began with the exploration of quantum systems. We believe that 
our paper, situated at the intersection of classical and quantum integrability 
and founded on a novel approach to the quantisation problem, will resonate with 
Decio's interests and legacy.

The problem of quantisation of dynamical systems has   a history spanning over 
a century. 
In 1925, Heisenberg put forward a new quantum theory, suggesting that the 
multiplication rules,
rather than the equations of classical mechanics, require modifications 
\cite{Heisenberg25}.
Almost immediately\footnote{The speed of publications and Dirac's reaction was 
astonishing!
Heisenberg's paper \cite{Heisenberg25} was submitted on 29th of July, published 
in September.
Dirac received a proof of Heisenberg's paper in August, submitted his paper 
containing a deep
development of Heisenberg's theory on 7th of November, which  was published on 
the 1-st of December 1925 \cite{Dirac25}.}
Dirac reformulated Heisenberg's ideas in mathematical form, introduced quantum 
algebra, and  noticed that in the  limit $\hbar\to0$ the Heisenberg commutators 
of quantum observables tend
to the Poisson brackets between the corresponding observables in the classical 
mechanics
$\hat{a}\hat{b}-\hat{b}\hat{a}\to i\hbar\{a,b\}$ \cite{Dirac25}.
In other words, non-commutative multiplication rules in the quantum theory can 
be regarded
as a deformation of commutative multiplication of smooth functions.
Nowadays there is enormous amount of papers, books and conferences devoted to 
deformation quantisation. In \cite{Dirac25} Dirac stated
that ``The correspondence between the quantum and classical theories lies not so 
much
in the limiting agreement when $\hbar\to 0$ as in the fact that the mathematical 
operations
on the two theories obey in many cases the same laws'' and raise the important 
issue of self-consistency
of the quantum multiplication rules and their  consistency with the equations of 
motion for finite value of the Plank constant $\hbar$.

Recently, AVM presented a new approach to the problem of quantisation 
\cite{Mik-20}.
It is suggested to commence with dynamical systems defined on a free associative 
algebra, i.e. with a free
associative mechanics. In this theory smooth functions on a phase space (or a 
Poisson manifold) are replaced
by elements of a free algebra, generated by the dynamical variables. Any 
finitely generated associative algebra, including Dirac's quantum algebra,
can be regarded as a quotient of a free algebra with an equivalent number of 
generators over a suitable two-sided ideal. The commutation rules of a  quantum  
theory  enables
one to swap positions of any two variables. In \cite{Mik-20} by quantisation it 
is understood a reduction of a system defined on a free associative algebra 
to the dynamical system on  a quotient algebra such that any two generators can 
be re-ordered using its multiplication rule.
In order to achieve the consistency (to solve the issue raised by Dirac) the 
ideal (the quantisation ideal)
should be invariant with respect to the derivation defined by the dynamical 
system.
The classical commutative case corresponds to the ideal 
generated by the commutators
of all dynamical variables. The  method of quantisation proposed in 
\cite{Mik-20}  does not appeal to a Poisson structure of the system, and 
therefore it enables to define a concept of non-deformation quantisation.
For example, the Volterra integrable lattice admits a deformation 
quantisation. Using the new method it is shown that  its cubic symmetry admits 
two different quantisations, and one of which is  
non-deformation. This approach has been developed further and 
applied to quantisation of the Volterra   hierarchy in \cite{cmw-22}. In 
particular it was shown that a periodic Volterra lattice with period three 
admits a bi-quantum structure which is a quantum analog of the corresponding 
bi-Hamiltonian structure.

The aim of our paper is to apply the approach proposed in \cite{Mik-20} to the 
problem of quantisation
of stationary flows of the KdV hierarchy, known as the Novikov equations 
\cite{B-Mikh-21}, \cite{Dickey-03},
\cite{Gelfand-Dikii-79}, \cite{Nov-74}. Novikov discovered that the stationary 
flows of the KdV equation is a completely integrable dynamical system, it 
possess a rich family of periodic
and quasiperiodic exact solutions which can be expressed in terms of Abelian 
functions  \cite{DubNov},\cite{Nov-74}.
Here we would like to emphasise that we study the problem of quantisation of  
finite dimensional systems of ordinary
differential equations and {\em not} of the field theory associated with partial 
differential equations of the KdV hierarchy.

In Section \ref{sec1} we give an explicit algebraic description of $N$-th 
Novikov equation and the corresponding finite
 hierarchy of symmetries in the form convenient for further generalisations. The 
$N$-the Novikov equation is
an ordinary differential equation of order $2N$. In Proposition \ref{prop3} it 
is shown that a complete set of $N$ first
integrals of the $N$-th Novikov equations can be explicitly presented in  terms 
of the coefficients of fractional powers
$L^{\frac{2k-1}{2}}$ of the Schr\"odinger operator $L=D^2-u,\ D=\frac{d}{dx}$.
The KdV hierarchy defines evolutionary derivations in  the graded algebra 
$\fA_0=\bbbC[u_0,u_1,u_2,\ldots]$ with the weights $|u_k|=k+2$, where 
$u_0=u,\,u_{k+1}=D(u_k)$. The commuting evolutionary derivations define a 
representation of algebraically independent variables $u_0,u_1,u_2,\ldots$ as 
smooth functions $u_k=u_k(t_1,t_3,\ldots)$  of graded variables 
$t_{2k-1},\ k\in\bbbN$ where the weight $|t_{2k-1}|=-2k+1$ and $t_1=x$.
We treat the $N$-th Novikov equation as a generator of a differential ideal 
$\fI_N$ in the graded ring $\fA=\cA[u,u_1,\ldots]$,
where $\cA$ is a commutative algebra of graded parameters $\alpha_4,\alpha_6, 
\ldots,\alpha_{2N+2}$ where the weights $|\alpha_{2n}|=2n$.
The Proposition \ref{prop2} shows that the KdV hierarchy induces the finite $N$  
hierarchy of integrable ordinary differential
polynomial equations on the quotient ring $\fA\diagup\fI_N$ which is called 
$N$-th Novikov hierarchy. Here by integrability we understand the existence of 
$N$ first integrals
and $N$ commuting symmetries, one of which is the $N$-th Novikov equation 
itself.

Let $\mathcal{B}$ be an associative $\mathbb{C}$-algebra with the unit $1$ and 
$\mathcal{M}$ be a complex linear space.
Let  $\varepsilon\colon\mathcal{B}\to\mathcal{M}$ be a linear map such that 
$\varepsilon(1)\neq 0$.
In this paper  we introduce the Frobenius-Hochschild algebra 
$FH(\mathcal{B},\mathcal{M})$. The name and notation are motivated
by the fact that the structure of a $FH(\mathcal{B},\mathcal{M})$-algebra 
$\mathcal{U}$ is given by a skew-symmetric quadratic form $\Phi$
on the $\mathcal{B}$-bimodule $\mathcal{U}$ with values in $\mathcal{M}$, and 
this form $\Phi$ is a $1$-cocycle in the cochain Hochschild
complex of the algebra $\mathcal{U}$. Partial cases of 
the Frobenius-Hochschild algebras are anti-Frobenius algebras. The latter was 
introduced and developed in connection with the associative Yang-Baxter 
equation, see \cite{Sok-20}.
In Section 1.2 we describe   properties of the 
$FH(\mathcal{B},\mathcal{M})$-algebra in the case $\mathcal{B}  = 
\mathcal{M}= \fA_0$.
In terms of the form $\Phi = \sigma(\cdot,\cdot)$ of this algebra, the first 
$N$-Novikov integrals are explicitly described
in the case of a commutative ring of polynomials $\fA_0$.

The KdV equations with non-commutative matrix variables were introduced in 
\cite{WadKam74}, \cite{CalDeg77}.
The KdV hierarchies on free associative algebras were studied in 
\cite{DorFok92}, \cite{EtGelRet97}, \cite{OS-98}, \cite{OW-00}, \cite{Sok-20}.
In Section \ref{sec2} we give a description of the integrable KdV hierarchy on 
a 
differential graded free associative algebra
$\fB_0=\bbbC\langle u, u_1, \ldots \rangle,\ D(u_k)=u_{k+1}$. Here, by 
integrability we understand the existence of an infinite hierarchy
of commuting symmetries, which are generators of symmetries of the 
non-commutative KdV equation.
There is a complete classification of integrable hierarchies of evolutionary 
non-commutative equations  \cite{OW-00}.
In particular, it was shown that the hierarchy of the KdV equation can be 
generated by a (non-local) recursion operator.
In the non-commutative case in order to define local conservation laws we need 
to introduce a linear space of functionals
with the values in the quotient linear space $\fB_0\diagup \big( {\rm 
Span}([\fB_{0},\fB_{0}])\oplus D(\fB_0)\big)$,
see \cite{DorFok92}, \cite{OS-98}, \cite{OW-00}.  Formal definitions of the  
$N$-th Novikov equation and its hierarchy of symmetries
are the same as in the commutative case. Namely, we take a stationary flow of a 
linear combination of the first $N$ members
of the KdV hierarchy with commuting constant coefficients $\alpha_{2n}\in\cA$, 
as a generator of the two-sided ideal
$\fI_N\subset\fB=\cA\langle u, u_1\ldots \rangle$. The $N$-Novikov hierarchy is 
defined as the canonical projection
of the   KdV hierarchy to the quotient ring $\fB_N=\fB\diagup \fI_N$ which is 
free over $\cA$ and finitely generated.
The first system of the hierarchy $\partial_{t_1}u_k=\cD(u_k), \ k=0,\ldots 
2N-1$ is the $N$-th Novikov equation itself, 
written in the form of a first order system where $\cD$ is the derivation of 
$\fB_N$ induced by $D$.
In contrast to the commutative case, the hierarchy of linearly independent 
symmetries is infinite.
The  case $N=1$ is already nontrivial. For $N=1$ the Novikov equation coincides 
with the (non-commutative) Newton equation
$ u_2=3u^2+8\alpha_4$ and in $\fB_1$ is represented by the first order system
\begin{equation} \label{N1}
\partial_{t_1}u=u_1,\quad \partial_{t_1}u_1=3u^2+8\alpha_4.
\end{equation}
Equation (\ref{N1}) admits an infinite hierarchy of commuting symmetries. First 
four of them are presented in Section 2.
 
A general definition of first integrals for equations on free associative 
algebra was discussed in \cite{miksok_CMP}.
First integrals for the non-commutative $N$-th Novikov equation and its  
hierarchy are introduced in Definition \ref{deffirst}.
In Section 2.2 we describe the properties of the 
$FH(\mathcal{B},\mathcal{M})$-algebra, where $\mathcal{B} = \fB_0$, and 
$\mathcal{M} = \fB_0/{\rm Span}([\fB_{0},\fB_{0}])$. First integrals of the 
non-commutative $N$-th Novikov  hierarchy are explicitly 
represented 
in terms of the form $\Phi = \sigma(\cdot,\cdot)$ of this algebra.
Using Lemma \ref{Lem-2}, we constructed  infinitely many algebraically 
independent first integrals 
for the non-commutative $N$-th Novikov equation and its hierarchy.

In Section \ref{sec3} we consider the quantisation problem for $N$-th Novikov 
equation following the method proposed in \cite{Mik-20}. Let $\cQ_N$ be a 
commutative graded ring of parameters
\[
 \cQ_N=\bbbC[\alpha_{2j+2}, q_{i,j}, q_{i,j}^\omega\,|\; 0\leqslant 
i<j\leqslant 
2N-1,\; 0\leqslant |\omega|<i+j+4]
\]
where $|q_{i,j}|=0$,\; $\omega=(i_{2N-1},\ldots,i_1,i_0)\in \bbbZ_\geqslant 
^{2N}$,\;
$|\omega|=(2N+1)i_{2N-1}+\cdots + 3i_1+2i_0$,\; 
$|q_{i,j}^\omega|=i+j+4-|\omega|$, and  $\fB_N(q)$ denotes the graded free 
associative  ring
$
 \fB_N(q)=\cQ_N \langle u_0,\ldots,u_{2N-1} \rangle 
$.
Having $N$-th Novikov equation on   $\fB_N(q)$,
we introduce a differential homogeneous two-sided ideal $\fQ_N\subset \fB_N(q)$ 
generated by the polynomials
\begin{equation}\label{dp-1i}
 p_{i,j} = u_iu_j - q_{i,j}u_ju_i + \sum\limits_{0 \leqslant|\omega|< i+j+4} 
q_{i,j}^\omega u^\omega,\quad
 0\leqslant i < j\leqslant 2N-1,\; q_{i,j} \neq 0.
\end{equation}
where
$u^\omega=u_{2N-1}^{i_{2N-1}}\cdots  u_{1}^{i_1} u_{0}^{i_0}$ are normally 
ordered monomials.
The ideal $\fQ_N$ is a \emph{ quantisation ideal} of the $N$-th Novikov equation
if the quotient algebra $\fC_N = \fB_N(q)/\fQ_N$ has a 
Poincar\'e--Birkhoff--Witt additive $\cQ_N$-basis of {\sl normally ordered 
monomials}
$u^\omega,\ \omega \in \bbbZ_\geqslant ^{2N}$, and $\fQ_N$ is invariant with 
respect to the derivation $\cD$.
It follows from $\cD(\fQ_N)\subset\fQ_N$ that the coefficients $q_{i,j}, 
q_{i,j}^\omega$
satisfy a system of algebraic equations. In particular, these equation imply 
that $q_{i,j}=1$ for all $0\leqslant i < j\leqslant 2N-1$ (Lemma \ref{Lem-5}).
In the cases $N=1,2,3$ and $4$ we have found out that the all structure 
constants $q_{i,j}^\omega$ of the quantisation ideals $\fQ_N$
can be parameterised by one parameter which we denote $\hbar$. In the case $N=1$ 
the computations are presented in full detail in Section \ref{sec32}.
In this case we have shown that the quantisation ideal for equation (\ref{N1}) 
is generated by the commutation relation $[u_1,u]=i\hbar$,
which coincides with Heisenberg's commutation relation in quantum mechanics 
\cite{Heisenberg25}, \cite{Dirac25}.
In the case $N=2$ we have shown (Proposition \ref{q2}) that  the quantisation 
ideal $\fQ_2$ is generated by six commutation relations
\[\begin{array}{l}\phantom{.}
 [u_i,u_j]=0\ \mbox{\rm for }\ i+j<3\ \mbox{\rm or }\ i+j=4; \\\phantom{.}
 [u_3,u ]=[u_1,u_2]=i\hbar,\quad [u_3,u_2]=10i\hbar u_0,
  \end{array}
\]
The quantum $N=2$ KdV hierarchy can be written in the Heisenberg form (Theorem 
\ref{T-1})
\begin{align*}
\partial_{t_1}u_k & =\; \fD(u_k)= \frac{i}{\hbar}[ \fh_{5,3},u_k] =
\begin{cases}
u_{k+1},\; 0\leqslant k\leqslant 2, \\
32\alpha_6 - 16\alpha_4u  + 5u_{1}^2 + 10u_{2}u  - 10u ^3, \; k=3;
\end{cases} \\
4\partial_{t_3}u_k & =\; \frac{i}{\hbar}[\fh_{5,5},u_k ] =
\cD^{k+1}(u_2 - 3u^2).
  \end{align*}
Here the Hamiltonian  $\fh_{5,3}\in \fC_2 $ for the Novikov equation coincides 
with the first integral of weight $8$ in the commutative case,
assuming that all monomials are normally ordered, while the Hamiltonian 
$\fh_{5,5}\in \fC_2 $ requires a quantum correction (Proposition \ref{hams2}).
These Hamiltonians commute with each other $[\fh_{5,3},\fh_{5,5}]=0$. We 
conclude Section \ref{sec3} by discussion of quantum ideals
for $N=3$ and $N=4$ and the hierarchy of quantum KdV equations in the Heisenberg 
form in the case $N=3$.

We emphasize, that the method of quantisation proposed in \cite{Mik-20} does not 
assume any Hamiltonian structure
of the noncomutative dynamical system, nevertheless we present the quantum 
equations in the Heisenberg form
$\partial_t u_k=\frac{i}{\hbar}[\fH,u_k]$ in Section \ref{sec3}.

\section{Novikov's equations and the corresponding finite KdV 
hierarchies.}\label{sec1}

\subsection{Lie algebra of evolutionary differentiations.}\text{}
Consider a graded  commutative differential polynomial algebra
 \begin{equation}\label{fAD}
  \fA_0=(\mathbb{C}[u_0,u_1,\ldots],\ D),
 \end{equation}
where $D$ is a derivation of $\mathbb{C}[u_0,u_1,\ldots]$ such that 
$D(u_k)=u_{k+1},\ k=0,1,\ldots$
In terms of grading we assume that the variables $u_k$ have weight $|u_k|=k+2$ 
and operator $D$ have weight $|D|=1$.
The variable $u_0$ will be often denoted as $u$.
The derivation $D$ can be represented in the form
\[
 D=X_{u_1}=\sum_{k=0}^\infty u_{k+1}\frac{\partial}{\partial u_k}.
\]

Derivations of $\fA_0$ form a Lie algebra ${\rm Der\,}\fA_0$ over $\bbbC$. A 
formal sum
\begin{equation}\label{derX}
 X= \sum_{n=0}^{\infty}f_n \frac{\partial}{\partial u_{n}},\qquad f_n\in\fA_{0},
\end{equation}
is a derivation in $\fA_{0}$. Its action $ X:\fA_{0} \mapsto\fA_{0} $ is well 
defined, since any element\\
$a\in \fA_{0} $ depends on a finite subset of variables, and therefore the sum 
$X(a)$ contains only a finite number
of non-vanishing terms. The $\mathbb{C}$ linearity and the Leibniz rule are 
obviously satisfied.
For example, partial derivatives $\frac{\partial}{\partial u_i},\ i=0,1,\ldots$, 
are
commuting derivations in $\fA_0$.

A  derivation  $X$ is said to be {\em evolutionary} if it commutes with the 
derivation $D$.
For an evolutionary derivation it follows from the condition $XD = DX$ that all 
coefficients
$f_n$ in (\ref{derX}) can be expressed as $f_n = D^n(f)$ in terms of one element 
$f\in\fA_{0}$, which is called
the {\em characteristic} of the evolutionary derivation. We will use notation
\begin{equation} \label{XF}
 X_f = \sum_{i=0}^\infty D^i(f) \frac{\partial}{\partial u_{i}}
\end{equation}
for the evolutionary derivation corresponding to the characteristic  $f$. The 
derivation $D$ is also evolutionary
$D=X_{u_1}$ with the characteristic $u_1$.

Evolutionary derivations form a Lie  subalgebra of the Lie algebra ${\rm 
Der\,}\fA_0$. Indeed,
\[\begin{array}{l}
   \alpha X_f+\beta X_g=X_{\alpha f+\beta g},\quad \alpha,\beta\in\mathbb{C} ,\\
\phantom{}   [X_f,X_g]=X_{[f,g]},
  \end{array}
\]
where $[f,g]\in\fA_{0} $ denotes the Lie bracket
\begin{equation}\label{bracket}
 [f,g]=X_f(g)-X_g(f),
\end{equation}
which is bi-linear, skew-symmetric and satisfying the Jacobi identity. Thus 
$\fA_{0} $ is a Lie algebra
with Lie bracket defined by (\ref{bracket}).

Let  $a(u ,\ldots,u_n)$ be a non-constant element of $\fA_0$. Then $X_f(a)$ can 
be represented by a finite sum
\begin{equation}\label{XFa}
 X_f(a)=\sum_{i=0}^n   \frac{\partial a}{\partial u_{i}}D^i(f)=a_*(f),
\end{equation}
where
\begin{equation}\label{astar}
a_* = \sum_{i=0}^n\frac{\partial a}{\partial u_i}D^i
\end{equation}
is the {\em  Fr\'echet derivative} of $a(u ,\ldots,u_n)$ and  $a_*(f)$ is the 
Fr\'echet derivative of $a$ {in the direction } $f$.
Using the Fr\'echet derivative  we can represent the Lie bracket (\ref{bracket}) 
in the form
\begin{equation}\label{LieF}
 [f\, ,\, g]=g_*(f)-f_*(g).
\end{equation}

An evolutionary derivation $X_f$ we identify with the partial differential 
equation 
\begin{equation}\label{eq}
 \partial_t (u)=f,\qquad f\in\fA_0.
\end{equation}

Following \cite{Sok-20} we define symmetries of (\ref{eq}).
\begin{defn}
 A dynamical system
\begin{equation}\label{syeq}
 \partial_\tau(u)=g,\qquad g\in\fA_0
\end{equation}
is called an inﬁnitesimal symmetry (or just symmetry for brevity) for 
(\ref{eq}) if (\ref{eq}) and (\ref{syeq}) are compatible.
\end{defn}
It is clear that equation (\ref{syeq}) is a symmetry of equation (\ref{eq}) iff 
$[X_f,X_g]=0$. 
By a symmetry we will also call the evolutionary derivation $\partial_\tau$ 
which commutes with $\partial_t$.

\subsection{Frobenius--Hochschild algebras.}\text{}

We shall assume that $u$ is a smooth function  $u = 
u(t_1,t_3,\ldots,t_{2k-1},\ldots)$ of
graded variables $t_{2k-1},\,k=1,2,\ldots$, where $|t_{2k-1}| = 1-2k$. The 
variable $t_1$ we will identity with $x$.
We use abbreviated notations for partial derivatives $\frac{\partial u}{\partial 
t_{2k-1}} = \partial_{t_{2k-1}}(u)$
and $\partial_x =\partial_{t_1}=D$. The grading weights  $|\partial_{t_{2k-1}}| 
= 2k-1$.

Let us define a differential operator of order $m$ as a finite sum of the form
\[
 A= \sum\limits_{i=0}^m a_iD^i,\; a_i\in \fA_0,\, a_m\neq 0
\]
where $D^0=1$ is the identity operator.

An operator $A$ is called  \emph{homogeneous} of weight $k$, if $|a_i|+i = k$ 
for all $i$. Differential operators act
naturally on the algebra $\fA_0$.

The set of differental graded operators
\[
 \fA_0[D]=\Big\{\sum\limits_{i=0}^m a_iD^i\,|\,  a_i\in \fA_0,\,a_m\neq 0,\,m 
\in \mathbb{Z}_{\geqslant 0}\Big\}
\]
and the set of graded formal differential series
\begin{equation}\label{fD}
\fA_0^D=\fA_0[D][[D^{-1}]]= \Big\{\sum\limits_{i\leqslant m} a_{i}D^{i}\,|\,  
a_i\in \fA_0,\, a_m\neq 0,\, m\in\bbbZ\Big\}
\end{equation}
are non-commutative associative algebras. In this algebra, multiplication is 
defined by the composition of series using the formula
\begin{equation}\label{fA-2}
bD^ka\,D^l = \sum_{i\geqslant 0}\binom{k}{i}ba^{(i)}D^{k+l-i}
\end{equation}
reflecting the Leibniz rule. Here $a^{(k)} = D^k(a)\in \fA_0$ and
\begin{equation}\label{binom}
\binom{k}{0} = 1,\quad \binom{k}{i} = \frac{k(k-1)\cdots(k-i+1)}{i!} = 
(-1)^i\binom{-k+i-1}{i},\; i>0.
\end{equation}
Obviously  $\fA_0[D]\subset\fA_0^D$ and $\fA_0[D]$ is a subalgebra in $\fA_0^D$.

For example,
\begin{align*}
 D^ka &=\, \sum_{i=0}^k \binom{k}{i}a^{(i)}D^{k-i},\, k\geqslant 0; \\
 D^{-1}a &=\, \sum_{i\geqslant 0}(-1)^i a^{(i)}D^{-(i+1)} = aD^{-1} - D(a)D^{-2} 
+ D^2(a)D^{-3} -
 D^3(a)D^{-4} + \cdots
\end{align*}

For any two elements $A,B\in\fA_0^D$ we have the commutator $[A,B] = AB - BA$.
For instance, for any $a\in \fA_0$, the formulas are fulfilled:\\ $[D,a] = 
D(a)$;\; $[D^{-1},a] = -D(a)D^{-2} + D^2(a)D^{-3} - \cdots$

\begin{defn}
For a formal series $A\in\fA_0^D$ the coefficient $a_{-1}$ of the term 
$a_{-1}D^{-1}$ is called the residue of this series $A$
and denoted by $\res\, A$.
\end{defn}

\begin{lem}\label{Lem-1} \text{}
\begin{enumerate}
  \item[1)] For any $B \in \fA_0^D$ and $a \in \fA_0$ we have $\res\,[a,B] = 0$.
  \item[2)] For any $a \in \fA_0$ and $B,C \in \fA_0^D$
\begin{equation}\label{f-1}
\res\,[aB,C] = \res\,[B,Ca].
\end{equation}
\end{enumerate}
\end{lem}
\begin{proof}
1) Let $B = \sum\limits_{k\leqslant m} b_{k}D^{k}$. Then
\[
[a,B] = -\sum_{k\leqslant m}\sum\limits_{i>0}\binom{k}{i}b_{k}a^{(i)}D^{k-i}.
\]
Therefore
\[
\res\,[a,B] = \binom{k}{k+1}b_{k}a^{(k+1)} = 0,\quad k+1>0.
\]

2) For any elements $A,B,C$ of any associative algebra, the identity
\begin{equation}\label{f-Jac}
[A,BC]+[B,CA]+[C,AB] = 0
\end{equation}
holds. Therefore, for $a \in \fA_0$ and $B,C \in \fA_0^D$
\begin{equation}\label{f-2}
[aB,C] = [B,Ca]+[a,BC].
\end{equation}
Applying the operator ``$\res$'' to \eqref{f-2} and using already proved 
statement 1), we obtain the proof of statement 2).
\end{proof}

Let $\mathcal{B}$ be some associative $\mathbb{C}$-algebra with the unit $1$ and 
$\mathcal{M}$ some complex linear space.
Let a linear mapping $\varepsilon\colon\mathcal{B}\to\mathcal{M}$ be given such 
that $\varepsilon(1)\neq 0$.

\begin{defn}\label{def-3}
An associative $\mathbb{C}$-algebra $\mathcal{U}$ with unit $1$ will be called a 
Frobenius--Hochschild algebra
over $(\mathcal{B},\mathcal{M})$ (briefly $FH(\mathcal{B},\mathcal{M})$-algebra) 
if:
\begin{enumerate}
  \item[i)] The algebra $\mathcal{B}$ is a subalgebra of $\mathcal{U}$, and 
hence $\mathcal{U}$ is a two-sided $\mathcal{B}$-module.
  \item[ii)] The bilinear mapping $\Phi(\cdot,\cdot)\colon 
\mathcal{U}\otimes_\mathbb{C}\mathcal{U}\to\mathcal{M}$ is defined such that:\\
  1) for any $A\in\mathcal{U}$ and $b\in\mathcal{B}$ we have $\Phi(A,b) = 0$;\\
  2) for any $A,B,C\in\mathcal{U}$ the relation
  \begin{equation}\label{f-3}
  \Phi(A,BC)+\Phi(B,CA)+\Phi(C,AB) = 0
  \end{equation}
  is satisfied.
\end{enumerate}
\end{defn}

\begin{lem}\label{Lem-3}
Let $\mathcal{U}$ be some $FH(\mathcal{B},\mathcal{M})$-algebra. Then the 
bilinear mapping\\
$\Phi(\cdot,\cdot)\colon \mathcal{U}\otimes_\mathbb{C}\mathcal{U}\to\mathcal{M}$
\begin{enumerate}
  \item[a)] is skew-symmetric, i.e. for any $A,B\in\mathcal{U}$ the equality 
$\Phi(A,B) = -\Phi(B,A)$ is true;
  \item[b)] defines a bilinear mapping 
$\mathcal{U}\otimes_{\mathcal{B}}\mathcal{U}\to\mathcal{M}$, i.e., for any 
$A,B\in\mathcal{U}$
  and $a\in\mathcal{B}$, the equality $\Phi(aA,B) = \Phi(A,Ba)$ is true.
\end{enumerate}
\end{lem}
\begin{proof}
Let us substitute $C=1$ in \eqref{f-3}. Then, according to item 1) of the 
Definition \ref{def-3}, we obtain a proof of assertion a).
If we substitute $C=a$ in \eqref{f-3} then, according to item 1) of the 
Definition \ref{def-3}, obtain a proof of assertion b).
\end{proof}

\begin{thm}\label{T-5}
The algebra $\fA_0^D$ is a $FH(\fA_0,\fA_0)$-algebra in which the bilinear form 
$\Phi = \sigma\colon\fA_0^D\otimes_{\fA_0}\fA_0^D \to \fA_0$
is uniquely given by the formula
\begin{equation}\label{F-DnbDm}
\sigma(D^n,bD^m) =
\begin{cases}
\binom{n}{n+m+1}b^{(n+m)},& \text{if }\, n+m\geqslant 0,\, nm<0,\\
0,& otherwise.
\end{cases}
\end{equation}
\end{thm}
\begin{proof}
Let $A = \sum\limits_{k\leqslant m}a_kD^k$. Then $aA\in\fA_0^D$ for any 
$a\in\fA_0$ and therefore the algebra $\fA_0^D$
is a left $\fA_0$-module with respect to the embedding  
$\varepsilon\colon\fA_0\to\fA_0^D\,:\, a\to aD^0$.
According to \eqref{fA-2}, the structure of the right $\fA_0$-module is given by 
the formula
\[
Aa = \sum_{j\leqslant m}\left( \sum_{i=0}^{m-j}\binom{j+i}{i}a_{j+i}a^{(i)} 
\right)D^j.
\]

According to \eqref{F-DnbDm}, we obtain $\sigma(\varepsilon(a),A) = 
\sigma(aD^0,A) = 0$. Thus item 1) of condition ii)
of the Definition \ref{def-3} has been verified.

The proof of item 2) of condition ii) is based on two lemmas, which are of 
independent interest:

\begin{lem}\label{L-I}
For any $A,B\in\fA_0^D$
\begin{equation}\label{f-I-1}
D\big(\sigma(A,B)\big) = \res\,[A,B].
\end{equation}
\end{lem}
\begin{proof}
Forms $D\big(\sigma(\cdot,\cdot)\big)$ and $\res\,[\cdot,\cdot]$ are bilinear so 
it suffices to proof the relation:
\begin{equation}\label{f-19}
D\big(\sigma(aD^n,bD^m)\big) = \res\,[aD^n,bD^m],\; a,b\in\fA_0,\; 
n,m\in\mathbb{Z}.
\end{equation}
According to the condition of Theorem \ref{T-5}, for any $a,b\in\fA_0$ we have
\[
\sigma(aD^n,bD^m) = \sigma(D^n,bD^ma).
\]
But according to item 2) of Lemma \ref{Lem-1}
\[
\res\,[aD^n,bD^m] = \res\,[D^n,bD^ma].
\]
Therefore, it suffices to proof the relation \eqref{f-19} in the case $a=1$. But 
in this case we have:
\[
\res\,[D^n,bD^m] = \binom{n}{n+m+1}b^{(n+m+1)} = D\big(\sigma(D^n,bD^m)\big).
\]
Lemma \ref{L-I} is proved.
\end{proof}

The monomials
\[
u^\xi = u_n^{i_n}\cdots u_0^{i_0},\; \xi = (i_n,\ldots,i_0),\; i_n>0,\; 
i_k\geqslant 0,\; k=0,\ldots,n-1,\; |u^\xi| = \sum\limits_{k=0}^{n}(k+2)i_k,
\]
form an additive basis of the graded algebra $\fA_0 = 
\mathbb{C}[u_0,u_1,\ldots]$. We will consider $\fA_0$ as a graded algebra
$\fA_0 = \mathbb{C}\oplus\widetilde{\fA}_0$, where $\widetilde{\fA}_0 = \oplus_m 
\fA_0^m$ and $\fA_0^m$ is a graded finite-dimensional
$\mathbb{C}$-linear subspace in $\fA_0$ with an additive basis $\{ 
u^\xi,\,|u^\xi| = m \}$.

For example, $\{u\},\, \{u_1\}$, and $\{u_2,u^2\}$ are the basises of the spaces 
$\fA_0^2,\, \fA_0^3$, and $\fA_0^4$, respectively.

The vectors of the space $\fA_0^m$ are called homogeneous polynomials of weight 
$m$. Let us introduce an ordering
of the multiplicative generators of the algebra $\fA_0$:
\[
u=u_0<u_1<\cdots<u_k<u_{k+1}<\cdots
\]
Then a strict order is defined in the monomial basis $\{ u^\xi \}$ of the space 
$\fA_0^m$ for each $m>0$.
This order is induced by the lexicographic order of the sequences $\xi$.

\begin{lem}\label{L-II}
For any $m>0$ the homomorphism $D\colon \fA_0 \to \fA_0$ defines a monomorphism 
$\fA_0^m \to \fA_0^{m+1}$.
\end{lem}
\begin{proof}
By definition, $D(1) = 0$ and $D(u_k) = u_{k+1},\, k=0,1,\ldots$ Therefore, the 
differentiation operator $D$
takes $\fA_0^m$ to $\fA_0^{m+1}$. Let $u^\xi\in\fA_0^m$ where $\xi = 
(i_n,\ldots,i_0),\; i_n\neq 0$.
Then $D(u^\xi) = i_nu_{n+1}u_n^{i_n-1}u^{\hat{\xi}} + 
u_n^{i_n}D(u^{\hat{\xi}})$.
The composition of linear homomorphisms
\[
\overline{D}\colon \fA_0^m \to \fA_0^{m+1} \to \fA_0^{m+1}\;:\; 
\overline{D}(u^\xi) = u^{\xi'} = i_nu_{n+1}u_n^{i_n-1}u^{\hat{\xi}}
\]
maps the ordered set of monomials $u^\xi\in\fA_0^m$ into the ordered set of 
monomials $u^{\xi'}\in\fA_0^{m+1}$.
We will show that this mapping is monotone and thus we obtain that the 
homomorphism $D$ is a monomorphism for $m>0$.

Let $\xi_1 = (i_{n_1},\ldots,i_{0_1}) > \xi_2 = (i_{n_2},\ldots,i_{0_2})$ where 
$i_{n_1}\neq 0$ and $i_{n_2}\neq 0$.
Then $n_1\geqslant n_2$. If $n_1 > n_2$ then $\xi_1' > \xi_2'$.  If $n_1 = n_2$ 
and $i_{n_1} > i_{n_2}$
then in this case it is also obvious that $\xi_1' > \xi_2'$. Finally, let $n_1 = 
n_2 = n$ and $i_{n_1} = i_{n_2} = i_n$.
Then there is a sequence $\zeta = (i_n,\ldots,i_k),\; 0<k\leqslant n$, such that
$\xi_1 = (\zeta,\eta_1)$ and $\xi_2 = (\zeta,\eta_2)$ where $\eta_1 > \eta_2$. 
In this case
$u^{\xi_1'} = \overline{D}(u^{\xi_1}) = \overline{D}(u^{\zeta})u^{\eta_1}$ and 
$u^{\xi_2'} = \overline{D}(u^{\xi_2}) = \overline{D}(u^{\zeta})u^{\eta_2}$
and therefore $u^{\xi_1'} > u^{\xi_2'}$. Lemma \ref{L-II} is proved.
\end{proof}

We now continue the proof of Theorem \ref{T-5}. It remains to prove that item 2) 
of condition ii) of Definition \ref{def-3} is satisfied,
i.e. that relation \eqref{f-3} is true.

Let $A,B,C \in \fA_0^D$. Take the residue $\res$ of the left side of equality 
\eqref{f-Jac}.
Then, according to Lemma \ref{L-I}, we obtain:
\[
D\big(\sigma(A,BC)+\sigma(B,CA)+\sigma(C,AB)\big) = 0.
\]
Since according to Lemma \ref{L-II} the operator $D$ is the monomorphism on 
non-constant series, we obtain
that relation \eqref{f-3} is true. Theorem \ref{T-5}  is proved.
\end{proof}

\begin{cor} \label{C-2} \text{}
\begin{enumerate}
  \item[1)] For any $a,b\in\fA_0$ we have
\begin{equation}\label{F-ab}
\sigma(aD^n,bD^m) =
\begin{cases}
\binom{n}{n+m+1}\sum\limits_{s=0}^{n+m}(-1)^s a^{(s)}b^{(n+m-s)},& \text{if }\, 
n+m \geqslant 0,\, nm<0, \\
0,& \mbox{otherwise}.
\end{cases}
\end{equation}
  \item[2)] For any $A \in \fA_0^D$ we have
  \begin{equation}\label{F-A}
  \sigma(D,A) = \res\, A.
  \end{equation}
  \item[3)] For any $A,B \in \fA_0^D$ we have
  \begin{equation}\label{F-AB}
  \sigma\big(D,[A,B]\big) = D\big(\sigma(A,B)\big).
  \end{equation}
\end{enumerate}
\end{cor}
\begin{proof}
Assertion 1) follows from Lemma \ref{Lem-3} and formulas \eqref{F-DnbDm} and 
\eqref{fA-2}.
Assertion 2) follows from formula \eqref{F-ab}.
Assertion 3) follows from formula \eqref{F-A} and Lemma \ref{L-I}.
\end{proof}

For $A = \sum\limits_{i\leqslant m}a_iD^i,\, a_m\neq 0$, we put\, $A = A_+ + 
A_-$\, where $A_+ = 0$\, if\, $m<0$,
and\, $A_+ = \sum\limits_{i=0}^m a_iD^i$\, if\, $m\geqslant 0$.

\begin{cor}\label{C-3} \text{}
\begin{enumerate}
  \item[1)] $\sigma(A,B) = \sigma(A_+,B_-) + \sigma(A_-,B_+)$.
  \item[2)]  Let $[A,B] = 0$. Then $\sigma(A,B) = 0$.
\end{enumerate}
\end{cor}
\begin{proof}
Assertion 1) follows from formula \eqref{F-ab}.
Assertion 2) follows from item 3) of Corollary \ref{C-2} and Lemma \ref{L-II}.
\end{proof}

\begin{thm}\label{T-sigma}
The form $\sigma(\cdot,\cdot)$ is given in terms of the operation $\res$ by the 
recursive formula
\begin{equation}\label{ff-sigma}
 \sigma(A,BD) = \sigma(DA,B) - \res\, AB
\end{equation}
with the initial condition $\sigma(A,bD) = -\res\, Ab$ for any $A\in\fA_0^D$ and 
$b\in\fA_0$.
\end{thm}
\begin{proof}
As was noted above, $\res\, A = \sigma(D,A)$. For the triple $(A,B,D)$, 
according to identity \eqref{f-3}, we obtain formula \eqref{ff-sigma}.
\end{proof}

For example:
\[
\sigma(A,bD^2) = \sigma(DA,bD)-\res\, AbD = -\res\, (DAb+AbD)
\]
for any $A\in\fA_0^D$ and $b\in\fA_0$.

\subsection{KdV hierarchy.}\text{}

Let us consider a homogeneous operator $L = D^2-u$, $|L| = 2$.
\begin{lem}\label{Lem-1-2}
A homogeneous formal series
\[
\cL = D + \sum_{n\geqslant 1} I_{1,n}D^{-n},\; |\cL| = 1,
\]
where  $I_{1,n}\in \fA_0$  are homogeneous polynomials of the weight $n+1$,
satisfies the equation $\cL^2 = L$ if and only if $I_{1,1} = -\frac{1}{2}u, \; 
I_{1,2} = \frac{1}{4}u_1$ and
\begin{equation}\label{I-1n}
2I_{1,n} + I'_{1,n-1} + \sum_{k=1}^{n-2} 
I_{1,k}\sum_{i=0}^{n-k-2}(-1)^i\binom{k+i-1}{i} I^{(i)}_{1,n-k-i-1} = 0,\; 
n\geqslant 3.
\end{equation}
\end{lem}
\begin{proof}
Consider the equation
\[
\left(D + \sum_{k\geqslant 1} I_{1,k}D^{-k}\right)\left(D + \sum_{q\geqslant 1} 
I_{1,q}D^{-q}\right) = D^2 - u.
\]
We obtain
\[
\sum_{q\geqslant 1} DI_{1,q}D^{-q} + \sum_{k\geqslant 1} I_{1,k}D^{-k+1} + 
\sum_{k\geqslant 1,q\geqslant 1}I_{1,k}D^{-k}I_{1,q}D^{-q} = -u.
\]
Using \eqref{fA-2}, we get \eqref{I-1n}.
\end{proof}

Formula \eqref{I-1n} allows to calculate the polynomials $I_{1,n},\, n\geqslant 
3$, recursively
\begin{multline*}
\cL = D - \frac{1}{2}uD^{-1} + \frac{1}{4}u_1D^{-2} - \frac{1}{8}(u_2+u^2)D^{-3} 
+ \frac{1}{16}(u_3+6uu_1)D^{-4} - \\
- \frac{1}{32}(u_4+14u_2u+11u_1^2+2u^3)D^{-5} + \ldots
\end{multline*}

Let us define a sequence of differential operators
\begin{equation}\label{fA-3}
A_{2k-1} = \cL_+^{2k-1} = D^{2k-1} - \frac{1}{2}(2k-1)uD^{2k-3} +\cdots + 
a_{2k-1},\,k=1,2,\ldots,
\end{equation}
where $a_{2k-1} = A_{2k-1}(1) \in \fA_0,\, |a_{2k-1}| = 2k-1$,
and homogeneous differential polynomials $\rho_{2k}\in \fA_0,\; |\rho_{2k}| = 
2k$,
\begin{equation}\label{fA-4}
\rho_0 = 1, \quad \rho_{2k} = \res\, \cL^{2k-1},\; k=1,2,\ldots
\end{equation}
Thus
\[
\cL^{2k-1} = A_{2k-1} + \rho_{2k}D^{-1} + \ldots, \; k>0.
\]
We have $a_1 = 0$ and $\rho_2 = -\frac{1}{2}u$.

Let be
\[
\cL^{2k-1} = A_{2k-1} +  \sum\limits_{n\geqslant 1} I_{2k-1,n}D^{-n},\; k>0.
\]
From the relation $\cL^{2k+1} = L\cL^{2k-1}$ we obtain
\begin{align}
\label{A-2k1}
A_{2k+1} &= (D^2-u)A_{2k-1} + I_{2k-1,1}D + (I_{2k-1,2} + 2I'_{2k-1,1}), \\
\label{I-kn}
I_{2k+1,n} &= I_{2k-1,n+2} + 2I'_{2k-1,n+1} + I''_{2k-1,n} - uI_{2k-1,n}.
\end{align}

\begin{cor}\label{Cor-1}
\begin{equation}\label{a-2k1}
a_{2k+1} = (D^2-u)(a_{2k-1}) + I_{2k-1,2} + 2\rho'_{2k},\, k\geqslant 1,
\end{equation}
\begin{equation}\label{rho-k}
\rho_{2k+2} = I_{2k-1,3} + 2I'_{2k-1,2} + \rho''_{2k} - u\rho_{2k}.
\end{equation}
\end{cor}

Let $J = \langle u,u_1,\ldots\rangle\subset\fA_0$ be the two-sided maximal ideal 
generated by $u,u_1,\ldots$
\begin{prop}
For $k\in \mathbb{N}$ the following formula holds:
\begin{equation}\label{rho-2n}
\rho_{2k+2} = - \frac{1}{2^{2k+1}}\left(u_{2k} + \ldots + 
(-1)^k\binom{2k+1}{k}u^{k+1}\right) = -\frac{1}{2^{2k+1}}(u_{2k} - 
\widehat{\rho}_{2k+2})
\end{equation}
where $\widehat{\rho}_{2k+2}\in J^2$.
\end{prop}
\begin{proof}
By definition, $\rho_{2n} = \res\, L^{\frac{2n-1}{2}},\, n\geqslant 1$. Since 
$[D,u] = u_1$, then to calculate the coefficient
at $u^n$, it is sufficient to calculate the coefficient at $x^{-1}$ of the 
series $f(x) = (x^2-a)^{\frac{2n-1}{2}}$
where $[x,a] = 0$. We have
\[
f(x) = x^{2n-1}\Big(1-x^{-2}a\Big)^{\frac{2n-1}{2}} = x^{2n-1}\left( 
1+\sum_{k\geqslant 1}(-1)^k \binom{\frac{2n-1}{2}}{k}x^{-2k}a^k \right).
\]
Therefore, the desired coefficient is
\[
(-1)^n \binom{\frac{2n-1}{2}}{n}a^n = (-1)^n \frac{1}{2^{2n-1}} 
\binom{2n-1}{n}a^n.
\]
Formula \eqref{I-1n} implies that
\[
I_{1,n} = (-1)^n \frac{1}{2^{n}}u_{n-1}\!\mod J^2, \; n>0.
\]
Using formulas \eqref{I-kn} and \eqref{rho-k}, we obtain by induction that
\[
\rho_{2k+2} = - \frac{1}{2^{2k+1}}u_{2k}\!\mod J^2, \; k\geqslant 0.
\]
\end{proof}

Examples:
\begin{align*}
\rho_4 &= -\frac{1}{8}(u_2-3 u^2),\qquad \rho_6=-\frac{1}{32}(u_4-10 u_2 u-5 
u_1^2+10 u^3),\\
\rho_8 &= -\frac{1}{128}(u_6-28 u_3 u_1-14 u_4 u-21 u_2^2+70 u_2 u^2+70 u_1^2 
u-35 u^4).
\end{align*}

It is easy to show that $[\cL^{ 2k-1},L]=0$ and therefore the commutator
\[
 [A_{2k-1},L]=[\cL^{ 2k-1}-(\cL^{ 2k-1})_-,L]= 2D(\rho_{2k})\in\fA_0,
\]
is the operator of multiplication on the function $2D(\rho_{2k})$.

\begin{defn}
The KdV hierarchy is defined as an infinite sequence of differential equations
\begin{equation}\label{fA-5}
\partial_{t_{2k-1}}(u) = -2D(\rho_{2k}),\,k\in \mathbb{N}.
\end{equation}
\end{defn}
Examples:
\[
 \begin{array}{rll}
\partial_{t_1}(u)&=&u_1  ,\\
4\partial_{t_3}(u)&=&u_3-6uu_1 ,\\
16\partial_{t_5}(u)&=&u_5-10uu_3-20u_1u_2+30u^2u_1 ,
 \end{array}
\]
and so on.

The partial derivatives  $\partial_{t_{2k-1}}$ can be extended to derivations of 
the algebra $\fA_0^D$
\[
\partial_{t_{2k-1}}(A) = \sum_{i\leqslant m}\partial_{t_{2k-1}}(a_i)D^i,\; 
\text{ where } A = \sum_{i\leqslant m} a_iD^i.
\]
Therefore the KdV hierarchy can be written in the form of Lax's equations
\begin{equation}\label{fA-6}
\partial_{t_{2k-1}}(L) = [A_{2k-1},L].
\end{equation}
It can be shown, that the derivations $\partial_{t_{2k-1}}$ commute with each 
other \cite{Sok-20}, and thus
the KdV hierarchy is a system of compatible equations.

It follows from  $\partial_{t_{2k-1}}(L) = \partial_{t_{2k-1}}(\cL^2) = 
\partial_{t_{2k-1}}(\cL)\cL +
\cL\partial_{t_{2k-1}}(\cL)$ and  \eqref{fA-6} that
\[
\partial_{t_{2k-1}}(\cL) = [A_{2k-1},\cL],
\]
and therefore,
\begin{equation}\label{fA-7}
\partial_{t_{2k-1}}(\cL^{2n-1}) = [A_{2k-1},\cL^{2n-1}],\; n,k\in \mathbb{N}.
\end{equation}

Let's put
\begin{equation}\label{sigma-nk}
\sigma_{2k-1,2n-1} = \sigma(\cL^{2k-1}_+,\cL^{2n-1}_-)\in \fA_0.
\end{equation}
According to Corollary \ref{C-3}, we obtain:
\begin{equation}\label{dkrhon}
\sigma_{2k-1,2n-1} = \sigma_{2n-1,2k-1}.
\end{equation}
Taking the residue from the equation \ref{fA-7}, we get
\begin{equation}\label{fA-8}
\partial_{t_{2k-1}}(\rho_{2n}) = D(\sigma_{2k-1,2n-1}).
\end{equation}

Thus  $\{\rho_{2n},\ n\in\mathbb{N}\}$ is a sequence of common conserved 
densities for the infinite  KdV hierarchy (\ref{fA-5}),
and $\sigma_{2k-1,2n-1}$ are homogeneous differential polynomials,    
$|\sigma_{2k-1,2n-1}|=2n+2k-2$.

On the algebra $\fA_0$ the evolutionary derivations   $\partial_{t_{2k-1}}$ are 
represented by commuting derivations
\begin{equation}\label{Dk}
 D_{2k-1}= -2\sum_{\ell=0}^\infty D^{\ell+1}(\rho_{2k})\frac{\partial}{\partial 
u_\ell}.
\end{equation}
In particular $D_1=D$,
\[
 D_3=\frac{1}{4}(u_3-6uu_1)\frac{\partial}{\partial 
u}+\frac{1}{4}(u_4-6uu_2-6u_1^2)\frac{\partial}{\partial u_1}+\cdots
\]

\subsection{The $N$-th Novikov hierarchy.}\text{}

Let us choose a positive integer $N$. Let $\mathcal{A} = 
\mathbb{C}[\alpha_4,\ldots,\alpha_{2N+2} ]$ be a graded algebra
of parameters and $ \fA = \mathcal{A}[u_0,u_1,\ldots]$.
We assume that $|\alpha_{2n}| = 2n,\, n\geqslant 2$, and $\alpha_{2n}$ are 
constants, meaning that
$D_{2k-1}(\alpha_{2n}) = 0$ for all $k\geqslant 1$ and $n\geqslant 2$.

Let us  define a symmetry $\partial_\tau$ of the KdV equation taking a linear 
combination
with constant coefficients of the first $N$ members of the KdV hierarchy  
\eqref{fA-5}
\begin{equation}\label{n0}
   \partial_{\tau}(u) = \partial_{t_{2N+1}}(u) + 
\sum_{k=1}^{N-1}\alpha_{2(N-k+1)}\partial_{t_{2k-1}}(u).
\end{equation}
Let us define a  polynomial
\begin{equation}\label{n1}
   F_{2N+2} =\rho_{2N+2} + \sum_{k=0}^{N-1}\alpha_{2(N-k+1)} \rho_{2k}.
\end{equation}
In (\ref{n1}) we assume that $\rho_0=1$ and $\alpha_{2N+2}$ is a constant 
parameter of weight $|\alpha_{2N+2}|=2N+2$.
The polynomial $F_{2N+2}$ (see \eqref{n1}) is homogeneous of weight $2N+2$. Let 
us restrict ourselves
with solutions of the KdV hierarchy which are invariant with respect to the 
symmetry \eqref{n0}. It implies that
\begin{equation}\label{Nov}
\rho_{2N+2} + \sum_{k=0}^{N-1}\alpha_{2(N-k+1)} \rho_{2k}=0.
\end{equation}
It follows from \eqref{rho-2n} that equation (\ref{Nov}) can be resolved with 
respect to the variable $u_{2N}$ and written in the form
\begin{equation}\label{Nov1}
 u_{2N} = f_{2N+2}(u_0,u_1,\ldots,u_{2N-2})
\end{equation}
where $f_{2N+2} = \widehat{\rho}_{2N+2} + 
2^{2N+1}\sum\limits_{k=0}^{N-1}\alpha_{2(N-k+1)}\rho_{2k} \in \fA$ is a 
homogeneous polynomial, $|f_{2N+2}| = 2N+2$.
Equation (\ref{Nov1}) is called $N$-th Novikov equation.
Since $\rho_{2n}\in\fA_0$, these equations depend linearly on 
$\alpha_4,\ldots,\alpha_{2N+2}$.

For example:
\[
 \begin{array}{llll}
  N=1:&  u_2 &=& 3u^2 + 8\alpha_4, \\
  N=2:&  u_4 &=& 10(u_2 - u^2)u + 5 u_1^2 - 16 \alpha_4 u + 32 \alpha_6, \\
  N=3:&    u_6 &=& 14(u_4  -  5u_2 u  +  5u_1^2)u + 28 u_1 u_3 + 21 
u_2^2 + 35 u^4 - 16 \alpha _4(u_2  -  3u^2)\\&& -& 64 \alpha_6 u + 128 \alpha_8.
 \end{array}
\]

Since $u_k = D^k(u) = u^{(k)}$, the $N$-th Novikov equation is an ordinary 
differential equation of the $2N$-th order
for the function $u = u(x)$.

Let $\fI_N = (F_{2N+2})\subset\fA$ be a differential ideal generated by the 
polynomial $ F_{2N+2} $ and the $D$ derivatives.
For any element of $\fA$ the canonical projection
\[
 \pi_N\,:\, \fA\mapsto \fA\diagup \fI_N
\]
is the result of the elimination of variables $u_k,\ k\geqslant 2N$, using 
equation \eqref{Nov1} and equation $u_{2N+k} = D^k(f_{2N+2})$ recursively.
\begin{prop}\label{prop2}
 The ideal $\fI_N$ is invariant with respect to evolutionary derivations 
$\partial_{t_{2k-1}},\, k\in\mathbb{N}$.
\end{prop}
\begin{proof}
Indeed, it follows from \eqref{XFa}, \eqref{dkrhon}, and \eqref{Dk} that
\begin{multline}\label{dtkN}
\partial_{t_{2k-1}}(F_{2N+2}) = 
\partial_{t_{2k-1}}\Big(\rho_{2N+2}+\sum\limits_{\ell=0}^{N-1}\alpha_{2(N-\ell+1
)} \rho_{2\ell}\Big) = \\
= \Big(\partial_{2N+1}+\sum\limits_{\ell=1}^{N-1}\alpha_{2(N-\ell+1)} 
\partial_{2\ell-1}\Big)(\rho_{2k})=-2(\rho_{2k})_*(D(F_{2N+2}))\subset\fI_N.
\end{multline}
\end{proof}

Commuting derivations  $D_{2k-1},\, 1\leqslant k \leqslant N$ (see (\ref{Dk})) 
induce  on $\fA/\fI_N$ the derivations
\begin{align*}
  \cD =&\; \cD_1 = \sum\limits_{\ell =0}^{2N-2}u_{\ell 
+1}\dfrac{\partial}{\partial u_\ell } + f_{2N}\dfrac{\partial}{\partial 
u_{2N-1}}\,,\\
  \cD_{2k-1} =&\, -2\sum\limits_{\ell=0}^{2N-1} 
\cD^{\ell+1}(\rho_{2k})\dfrac{\partial}{\partial u_\ell},\quad 1\leqslant k 
\leqslant N.
\end{align*}
In $\mathbb{C}^{2N}$ there are $N$ compatible systems of $N$ ordinary 
differential equations
\begin{equation}  \label{NovN}
 \partial_{t_{2k-1}}(u_s) = \cD_{ {2k-1}}(u_s) = -2\cD^{s+1}(\rho_{2k}),\quad 
s=0,\ldots,N-1,\; k =1,\ldots,N ,
\end{equation}
which we will call $N$-th Novikov hierarchy. In this case, the parameters 
$\alpha_{2k}$ are assumed to be fixed complex numbers.
In the hierarchy (\ref{NovN}), system with $k=1,\, s=0,\ldots,N-1$ represents 
the $N$-th Novikov equation \eqref{Nov1} as a first order system of $2N$ 
ordinary differential equations.

\begin{prop}\label{prop3}
The $N$-th Novikov equation  possesses $N$   first integrals
\begin{equation}\label{firstH}
H_{2n+1,2N+1}=\sigma_{2n+1,2N+1}+\sum_{k=1}^{N-1}\alpha_{2N-2k+2} 
\sigma_{2n+1,2k-1},\quad n=1,\ldots,N.
\end{equation}
The polynomials $H_{2n+1,2N+1}$ are homogeneous of weight $|H_{2n+1,2N+1}| = 
2N+2n+2$.
\end{prop}
\begin{proof}
It follows from (\ref{dkrhon}) that
\[
 \partial_{t_{2n-1}}(F_{2N+2})=D(H_{2n+1,2N+1})
\]
where
\begin{equation}\label{N-Int}
 H_{2n+1,2N+1}=\sigma_{2n+1,2N+1}+\sum_{k=1}^{N-1}\alpha_{2N-2k+2} 
\sigma_{2n+1,2k-1}.
\end{equation}
Thus,  it follows from (\ref{dtkN}) that $D(H_{2n+1,2N+1}) = 
\partial_{t_{2n+1}}(F_{2N+2})\in \fI_N$
and thus vanish in $\fA/ \fI_N$. For $n=1,\ldots,N$ one can check that 
$H_{2n+1,2N+1}\not \in \fI_N$,
and thus $H_{2n+1,2N+1}$ is a first integral of Novikov's equation. Moreover the 
first integrals
$H_{2n+1,2N+1},\ n=1,\ldots,N$, are algebraically independent.
\end{proof}

It is known that  $N$-th Novikov equation and equations of the $N$-th Novikov  
hierarchy can be reduced to integrable
Hamiltonian systems \cite{B-Mikh-21}.
It  follows from the general theory of integrable Hamiltonian systems that the 
derivations $\cD_{t_{2s-1}}$
with $s>N$ in  $\fA/ \fI_N$ are dependent, since they are linear combinations of 
$\cD_{t_{2k-1}},\ k =1,\ldots,N$, i.e.
\begin{equation}\label{finiteD}
\cD_{t_{2s-1}}=\sum\limits_{k=1}^N a_k\cD_{t_{2k-1}}
\end{equation}
with coefficients $a_k\in\bbbC[\alpha_4,\ldots,\alpha_{2N+2},H_{2N+1,3},\ldots, 
H_{2N+1,2N+1}]$ where $H_{2n+1,2N+1}$,\, $n=1,\ldots,N$,
are first integrals of the $N$-th Novikov equation. Using (\ref{firstH}) one can 
find the polynomials  $H_{2n+1,2N+1},\ n>N$,
which are also first integrals (it follows from the proof of Proposition 
\ref{prop3}), but they are algebraically dependent
with $ H_{3,2N+1},\ldots, H_{2N+1,2N+1}$.

For example:

${\bf N=1}$: The $N=1$ Novikov equation coincides with the Newton equation 
$\partial_{t_1}^2u=3u^2+8\alpha_4$.
\begin{align*}
  \partial_{t_1}(u) &=\, u_{1}; \\
  \partial_{t_{1}}(u_1) &=\, 3u^2 + 8\alpha_4,
\end{align*}
and according Proposition \ref{prop3} we get one first integral
\begin{equation}\label{H42}
 H_{3,3}=-\frac{3}{16}  \left(\frac{1}{2}u_1^2- u^3-8 \alpha _4 u\right).
\end{equation}

${\bf N=2}$: The hierarchy consists of two compatible systems in which the first 
one is the $N=2$ Novikov equation
\begin{align*}
  \partial_{t_1}(u_s) &=\, u_{s+1},\quad s=0,1,2; \\
  \partial_{t_{1}}(u_3) &=\, 32 \alpha_6-10 u^3-16 \alpha_4 u+10 u_2 u+5 u_1^2; 
\\
  4\partial_{t_{3}}(u_s) &=\, \cD^{s+1}(u_2-3u^2),\quad s=0,1,2,3\, .
\end{align*}
Proposition \ref{prop3} give us two first integrals
\begin{eqnarray}\label{H62}
 H_{3,5} &=& -\frac{3}{128}  \left(5 u^4+16 \alpha _4 u^2-64 \alpha _6 u-10 
u_1^2 u-u_2^2+2 u_3 u_1\right),\\ \label{H64}
 H_{5,5} &=& \frac{5}{512}  (24 u^5+64 \alpha _4 u^3-20 u_2 u^3-192 \alpha _6 
u^2-30 u_1^2 u^2-32 \alpha _4 u_2 u + \\
 &+& 16 \alpha _4 u_1^2+64 \alpha _6 u_2+4 u_2^2 u+12 u_3 u_1 u-u_3^2-2 u_2 
u_1^2 ).\nonumber
\end{eqnarray}

Obviously $H_{2n+1,2N+1}$ (see \eqref{firstH} in case $N=2$) are first integrals 
for any $n$, but they are algebraically dependent with $H_{3,5},H_{5,5}$
\begin{align*}
  H_{7,5} &=\, \frac{7 }{6}(3\alpha_6^2 - 2\alpha_4 H_{3,5}), \\
  H_{9,5} &=\, -3 \alpha_6 H_{3,5} - \frac{9 \alpha_4}{5}H_{5,5}, \\
  H_{11,5} &=\, \frac{11}{90} \Big(-45\alpha_4 \alpha_6^2 - 18\alpha_6 H_{5,5} + 
30\alpha_4^2 H_{3,5} + 5 H_{3,5}^2\Big),
\end{align*}
and so on.

${\bf N=3}$: The hierarchy consists of three compatible systems. The first one 
is the $N=3$ Novikov equation, the rest are its commuting symmetries:
\[
\begin{array}{rcll}
  \partial_{t_1}(u_s) &=& u_{s+1}, & s=0,\ldots,4; \\
  \partial_{t_{1}}(u_5) &=& 128\alpha_8 + 35 
u^4 + 48\alpha_4u^2 - 70u_2u^2 - 64\alpha_6u &\\& -& 16\alpha_4u_2 - 
  70u_1^2u+ 14u_4u + 21u_2^2 + 28u_1u_3;& \\
 4 \partial_{t_{3}}(u_s) &=& \cD^{s+1}(u_2-3u^2),& s=0,\ldots,5\, ; \\
  16\partial_{t_{5}}(u_s) &=& \cD^{s+1}(u_4-5u_1^2-10 uu_2+10u^3+16\alpha_4 
u),& s=0,\ldots,5\,.
\end{array}
\]
It follows from Proposition \ref{prop3} that there are three common first 
integrals of this systems
\begin{multline}\label{H37}
H_{3,7} = \frac{3}{2^9} \big(14u^5 + 32\alpha_4u^3 - 64\alpha_6u^2 - 70u_1^2u^2 
+ 256\alpha_8u - 16\alpha_4u_1^2 - 14u_2^2u + \\
+ 28u_3u_1u - u_3^2 + 28u_2u_1^2 + 2u_4u_2 - 2u_5u_1\big);
\end{multline}
\begin{multline}\label{H57}
H_{5,7} = -\frac{5}{2^{11}} \big(70u^6 + 144\alpha_4u^4 - 70u_2u^4 - 
256\alpha_6u^3 - 280u_1^2u^3 - 96\alpha_4u_2u^2+\\ + 768\alpha_8u^2 
- 14u_2^2u^2 + 168u_3u_1u^2 + 128\alpha_6u_2u + 16\alpha_4u_2^2 - 
64\alpha_6u_1^2-\\ - 256\alpha_8u_2 - 20u_3^2u 
+ 140u_2u_1^2u + 12u_4u_2u - 12u_5u_1u - 35u_1^4 - 2u_2^3 -\\- u_4^2 - 
36u_3u_2u_1 
+ 12u_4u_1^2 + 2u_5u_3\big);
\end{multline}
\begin{multline}\label{H77}
H_{7,7} = \frac{7}{2^{13}} \big(300u^7 + 576\alpha_4u^5 - 700u_2u^5 - 
960\alpha_6u^4 - 1050u_1^2u^4 + 70u_4u^4-\\ - 960\alpha_4u_2u^3 
+ 2560\alpha_8u^3 + 420u_2^2u^3 + 560u_3u_1u^3 + 96\alpha_4u_4u^2 + 
1280\alpha_6u_2u^2-\\ - 100u_3^2u^2 + 1400u_2u_1^2u^2  
- 80u_4u_2u^2 - 60u_5u_1u^2 + 256\alpha_4u_2^2u - 192\alpha_4u_3u_1u-\\ - 
128\alpha_6u_4u - 2560\alpha_8u_2u + 16\alpha_4u_3^2  
+ 192\alpha_4u_2u_1^2 - 32\alpha_4u_4u_2 - 64\alpha_6u_2^2 + \\+
128\alpha_6u_3u_1 - 1280\alpha_8u_1^2 + 256\alpha_8u_4 - 20u_2^3u + 4u_4^2u
- 360u_3u_2u_1u - 20u_4u_1^2u +\\+ 20u_5u_3u - 410u_2^2u_1^2 - u_5^2 + 
20u_3u_1^3 + 2u_4u_2^2 - 4u_4u_3u_1 + 40u_5u_2u_1\big).
\end{multline}

\section{KdV hierarchy and Novikov equations on free associative 
algebra.}\label{sec2}

\subsection{KdV hierarchy on free associative algebra.}\label{sec2-1} \text{}

It is well known that the KdV equation and its hierarchy can be defined on a 
free differential algebra
$ \fB_{0} = (\mathbb{C}\langle u_0,u_1,\ldots\rangle,D)$ with infinite number of 
noncommuting variables (see for example, \cite{OW-00}, \cite{Sok-20}).
Algebra $\fB_{0}$ has monomial additive basis $\{u_\xi=u_{i_1}u_{i_2}\ldots 
u_{i_m}\,|\,i_k\in \bbbZ_{\geqslant0},\ m\in\bbbN\}$. It is graded algebra
\begin{equation}\label{B0nm}
 \fB_0=\bbbC\bigoplus\limits_{n\geqslant2}\fB_{0,n},
\end{equation}
induced by the grading of the variables
$u_k,\, |u_k|=k+2$ for any $k\geqslant 0$ and therefore $|u_\xi|=i_1+\cdots 
+i_m+2m=n$. Here $\fB_{0,n}$ is a finite dimensional space 
$\fB_{0,n}={\rm Span}_\bbbC \langle u_\xi\, ;\, |u_\xi|=n\rangle$.

The construction of the hierarchy is similar to the commutative case, although 
one has to take care on the order of the variables,
since $u_k\cdot u_s\ne u_s\cdot u_k$ if $k\ne s$. Starting with the operator 
$L=D^2-u$, one can find its square root
by the formula $\cL = D + \sum\limits_{n\geqslant 1} I_{1,n}D^{-n}$, where 
$I_{1,n}\in \fB_0$ are non-commutative polynomials.
It follows from the proof of Lemma 2 that formula \eqref{I-1n} for the recursive 
calculation of the polynomials $I_{1,n}$
is also applicable in the case of a free associative algebra $\fB_0$.

Now the initial segment of the series $\cL$ has the form
\[
\cL = D - \frac{1}{2}uD^{-1} + \frac{1}{4}u_1D^{-2} - \frac{1}{8}(u_2+u^2)D^{-3} 
+ \frac{1}{16}(u_3+2u_1u+4uu_1)D^{-4} + \cdots
\]
Similarly to the commutative case, we introduce fractional powers $\cL^{2k-1}$ 
and polynomials $\varrho_{2k}=\res\,\cL^{2k-1}$.
From the identity $\cL^{2k+1} = L\cL^{2k-1}$ follows a formula for the recursive 
calculation of the polynomials $\varrho_{2k+2}$.
It follows from the proof of the formulas \eqref{I-kn} and \eqref{rho-k} that 
they are applicable in the case
of a free associative algebra $\fB_0$. However, expressions for 
$\varrho_{2k}\in\fB_0$ are different from expressions for $\rho_{2k}\in\fA_0,\, 
k\geqslant 3$,
\begin{equation}\label{qden-1}
 \varrho_2 = -\frac{1}{2}u;\quad  \varrho_4 = \frac{1}{8}(3u^2-u_2);\quad  
\varrho_6 = \frac{1}{32} \big(5(u_2u+uu_2)+5u_1^2-10u^3-u_4\big);
\end{equation}
\begin{multline}\label{qden-2}
 \varrho_8 = 
\frac{1}{128}\big(7(u_4u+uu_4)+14(u_3u_1+u_1u_3)+21u_2^2-21(u_2u^2+u^2u_2)- \\-
28uu_2u
 -28(u_1^2u+uu_1^2)-14u_1uu_1+35u^4-u_6\big);
\end{multline}
and so on.

\begin{defn}
The compatible system of equations on the free associative algebra $\fB_0$
\begin{equation}\label{kdv_hn}
\partial_{t_{2k-1}}(u) = -2D(\varrho_{2k})
\end{equation}
is called the  KdV hierarchy (similar to the commutative case \eqref{fA-5}).
\end{defn}
Equations of the KdV hierarchy define the commuting evolutionary derivations 
$D_{2k-1}$ of $\fB_0$. Their  action on the variables $u_n$ is given by
\[
 D_{2k-1}(u_n)=-2 D^{n+1}(\varrho_{2k})
\]
and it can be extended to $\fB_0$  by the linearity and the Leibniz rule.

We have
\begin{equation}\label{ukn-non}
  \partial_{t_{2k-1}}(\varrho_{2n}) = \partial_{t_{2n-1}}(\varrho_{2k}) = \res\, 
[\cL^{2n-1}_+,\cL^{2k-1}_-] = \res\, [\cL^{2k-1}_+,\cL^{2n-1}_-].
\end{equation}

In the non-commutative case $\res\,[A,B]$ is not any more in the image of the 
derivation $D$ and Lemma \ref{L-I} should be modified.
Let us introduce the algebra
\begin{equation*}
\fB_0^D = \Big\{A = \sum\limits_{i\leqslant m} a_{i}D^{i}\,|\,  a_i\in 
\fB_{0,|a_m|+m-i},\, a_m\neq 0,\, m\in\bbbZ\Big\} = \underset{k\in 
\mathbb{Z}}{\oplus}\fB_{0,k}^D\,.
\end{equation*}

\begin{defn}\label{Lem-2}
Let us introduce the homogeneous skew-symmetric bilinear over $\mathbb{C}$ form
\[
\widehat{\sigma}(\cdot,\cdot)\colon \fB_0^D\otimes\fB_0^D \rightarrow 
\fB_0,\quad |\widehat{\sigma}(A,B)| = |A|+|B|,
\]
such that for $ n,m\in \mathbb{Z}$
\begin{equation}
\widehat{\sigma}(aD^n,bD^m) =
\begin{cases} \label{FH-1}
\frac{1}{2}\binom{n}{n+m+1}\!\sum\limits_{s=0}^{n+m}(-1)^s 
(a^{(s)}b^{(n+m-s)}+b^{(n+m-s)}a^{(s)}), \text{if } n+m \geqslant 0,\, 
nm<0,\\[2pt]
0,\qquad \mbox{otherwise}.
\end{cases}
\end{equation}
\end{defn}

\begin{lem}\label{Cor-3}
We have $\res\,[A,B] = D\big(\widehat{\sigma}(A,B)\big) - 
\text{$\Delta$}(A,B)$\, for any\, $A,B \in \fB_0^D$ \, where
\[
\Delta(aD^n,bD^m) = \frac{1}{2}\binom{n}{n+m+1}\Big([a,b^{(n+m+1)}] + 
(-1)^{n+m}[b,a^{(n+m+1)}]\Big).
\]
\end{lem}
\begin{proof}
The statement of this lemma is verified by directly calculating the value of \\
$\Delta(aD^n,bD^m)$.
\end{proof}

\begin{cor}\label{Cor-2}
Let
\[
A_{2k-1} = \sum_{i=0}^{2k-1} a_{2k-1-i} D^i, \qquad \mathcal{L}_-^{2n-1} = 
\sum_{j\geqslant 1} I_{2n-1,j}D^{-j},\; n\geqslant 1.
\]
Then
\[
\res[A_{2k-1},\mathcal{L}_-^{2n-1}] = D(\widehat{\sigma}_{2k-1,2n-1}) - 
\Delta_{2k,2n-1}
\]
where
\begin{equation}\label{sigma}
\widehat{\sigma}_{2k-1,2n-1} = \frac{1}{2}\sum_{i=1}^{2k-1}\sum_{j=1}^i 
\binom{i}{i-j+1}\sum_{s=0}^{i-j} (-1)^s 
\left(a_{2k-1-i}^{(s)}I_{2n-1,j}^{(i-j-s)}
+ I_{2n-1,j}^{(i-j-s)}a_{2k-1-i}^{(s)} \right),
\end{equation}
\begin{equation}\label{Delta}
\Delta_{2k,2n-1} = \frac{1}{2}\sum_{i=1}^{2k-1}\sum_{j=1}^{i} \binom{i}{i-j+1} 
\left( [a_{2k-1-i},I_{2n-1,j}^{(i-j+1)}] + (-1)^{i-j} 
[a_{2k-1-i}^{(i-j+1)},I_{2n-1,j}]\right).
\end{equation}
\end{cor}

In the non-commutative case the definition of densities of local conservation 
laws has to be modified, since
\[
 \partial_{t_{2k+1}}(\varrho_{2n}) \in{\rm Span}[ \fB_{0}, \fB_{0}]\oplus 
D(\fB_0) .
\]
Here ${\rm Span}[\fB_{0},\fB_{0}]$ is a linear subspace generated by all 
commutators of elements from $\fB_{0}$.

A $\mathbb{C}$-linear space of functionals is defined as
$ \fB_{0}^\sharp= \fB_{0}/ \big( {\rm Span}[\fB_{0},\fB_{0}]\oplus 
D(\fB_0)\big)$, see \cite{DorFok92}, \cite{OS-98}, \cite{OW-00}. The polynomials
$\varrho_{2n}$ as elements of $\fB_{0}^\sharp$ are constants of motion of the 
nonabelian KdV hierarchy (\ref{kdv_hn}).

\subsection{Frobenius--Hochschild algebras over free associative 
algebra.}\label{sec2-2} \text{}

Denote by $\xi_k = (j_1,\ldots,j_k)$ sequences of non-negative integers of 
length $k\geqslant 1$. Let $u_{\xi_k} = u_{j_1}\cdots u_{j_k}$.
We obtain $|u_{\xi_k}| = 2k+ \sum\limits_{s=1}^k j_{s}$.

We will consider $\fB_{0}$ as a graded algebra $\fB_{0} = 
\mathbb{C}\oplus\widetilde{\fB}_{0}$, where $\widetilde{\fB}_{0} = 
\oplus_m\fB_{0,m},\, m\geqslant 2$,
and $\fB_{0,m}$ is a graded finite-dimensional $\mathbb{C}$-linear space with 
an additive lexicographically by indices ordered monomial basis
$\{ u_{\xi_k},\, |u_{\xi_k}| = m \}$.

For example, $\{ u_3,u_1u,uu_1 \}$ is a monomial and lexicographically ordered  
basis $ 
u_3\succ u_1u\succ uu_1$ in $\fB_{0}^5$.

Let $\xi_k = (j_1,\ldots,j_k)$ be a  multindex of the monomial $u_{\xi_k}$ and  
$T_k$ be a generator of the cyclic permutation group of order $k$: 
$T_1(\xi_1) = \xi_1$ and $T_k(\xi_k) = (j_2,\ldots,j_k,j_1),\, k\geqslant 2$. 
Let us denote by $ T(\xi_k)$ the maximal index set in $T(\xi_k) =\max _\succ\{ 
\xi_k,T_k(\xi_k),\ldots,T_k^{k-1}(\xi_k) \}$ with respect to the lexicographic 
ordering. 
We define the linear homomorphism $\cT\colon \fB_{0,m}\to \fB_{0,m}$ by its 
action  on the monomial basis elements
$\cT(u_{\xi_k}) = u_{T(\xi_k)}$.

For example, $\cT(u_1u) = \cT(uu_1) = u_1u$.

\begin{prop}\label{prop-18}
Homomorphism  $\cT\colon \fB_{0}\to \fB_{0}$ is a projector
such that \\$\ker\cT  = {\rm Span}[\fB_{0},\fB_{0}]$ and
  $\I\cT \simeq \fB_{0}^\natural = \fB_{0}/{\rm Span}[\fB_{0},\fB_{0}]$.

\end{prop}
\begin{proof}
It follows directly from the definition that  
$\cT=\cT^2$.
The properties of $\cT$ follow from the following facts:
\begin{enumerate}
  \item $[u_{\xi_k'},u_{\xi_s''}] = u_{\xi_{k. s}}-u_{T^k_{k+s}(\xi_{k.s})}$, 
where $\xi_{k.s} = (\xi_k',\xi_s'')$ is the concatination of  $\xi_k'$ and 
$\xi_s''$;
  \item $u_{\xi_k} - u_{T(\xi_k)} = [u_{j_1},u_{j_2}\cdots u_{j_k}]$ for $\xi_k 
= (j_1,\ldots,j_k),\, k>1$.
\end{enumerate}
\end{proof}
It follows from Proposition \ref{prop-18} that the projector $\cT$ gives the 
splitting of the exact sequence
\begin{equation}\label{sec19}
  0 \rightarrow {\rm Span}[\fB_{0},\fB_{0}] \rightarrow \fB_{0} 
{\longrightarrow} \fB_{0}^\natural \rightarrow 0.
\end{equation}
It enables us to identify the element $\cT(b),\ b\in\fB_0$ with its canonical 
projection in $\fB_0^\natural$.

\begin{thm}\label{T-19}
The algebra $\fB_{0}^D$ is the $FH(\fB_{0},\fB_{0}^\natural)$-algebra in which 
the bilinear form
$\Phi = \overline{\sigma}\colon \fB_{0}^D\otimes_{\fB_{0}}\fB_{0}^D \to 
\fB_{0}^\natural$ is uniquely given by the formula
\begin{equation}\label{FH-2}
\overline{\sigma}(D^n,bD^m) =
  \begin{cases}
    \binom{n}{n+m+1}\cT({b}^{(n+m)}), & \mbox{if } n+m\geqslant 0,\,nm<0, \\
    0, & \mbox{otherwise}.
  \end{cases}
\end{equation}
\end{thm}
\begin{proof}

Let us first explain the expression $\cT({b}^{(n+m)})$. The space ${\rm
Span}[\fB_{0},\fB_{0}]$ is closed under the differentiation $D$.
Therefore, the operator $D$ on $\fB_{0}$ uniquely determines the linear operator
$\overline{D}:\fB_{0}\mapsto\fB_{0}^\natural$, and for any $b\in \fB_{0}$
the formula $\overline{D}( b) = \cT(D(b))$ holds.

Let $\xi_k = (j_1,\ldots,j_k),\, k\geqslant 1$, and $ T(\xi_k) = 
(j_{1,*},\ldots,j_{k,*})$.
If $  \xi_k = (j_1,\ldots,j_1) = (j_1)^k,\, k\geqslant 1$, then $\overline{D}( 
u_{\xi_k}) = ku_{j_1+1}u_{j_1}^{k-1}$. If there are at least two distinct 
elements in the set $\xi_k$, 
then $\overline{D}( u_{\xi_k})$ is a  sum of monomials with the leading 
monomial 
$u_{j_1+1,*}u_{j_2,*}\cdots u_{j_{k,*}}$.
Thus, in a strictly ordered basis, the homomorphism $\overline{D}\colon  
{\fB}_{0,m}   \to  {\fB}_{0,m+1}^\natural$ is given
by an upper triangular matrix with a non-zero diagonal. It induces  the
monomorphism $\overline{D}\colon  
{\fB}_{0,m}^\natural   \to  {\fB}_{0,m+1}^\natural$.
Following the proof of Theorem \ref{T-5} and using Lemma \ref{Cor-3}, it is easy 
to complete the proof of Theorem \ref{T-19}.

\end{proof}

\begin{cor}
The algebra $\fB_{0}^D$ is the $FH(\fB_{0},\fB_{0}^\natural)$-algebra with
the bilinear form $\widehat{\sigma}(A,B)$     (\ref{sigma}).
\end{cor}
\begin{proof}
 It follows immediately from Theorem \ref{T-19} and the fact that \\
$\widehat{\sigma}(A,B)-\overline{\sigma}(A,B)\in   {\rm Span}[ \fB_{0}, 
\fB_{0}]$ for any $A,B\in \fB_{0}^D$.
\end{proof}

\subsection{Non-commutative $N$-th Novikov equation and its hierarchy.} 
\label{sec2-3}  \text{}

Let  $\fB=(\cA\langle u,u_1,\ldots \rangle,D)$ be the differential ring, where  
$\cA = \bbbC[\alpha_4,\ldots],\ D(u_k)=u_{k+1},\
D(\alpha_{2n})=0,$ and $ \alpha_{2n}$ are commuting parameters, i.e centre 
elements of $\fB$.
Similar to the commutative case, we fix a positive integer $N$ and define a 
homogenous polynomial $\fF_{2N+2}\in\fB$:
\begin{equation}\label{nn1}
   \fF_{2N+2} = \varrho_{2N+2} + \sum_{k=0}^{N-1}\alpha_{2(N-k+1)} \varrho_{2k}.
 \end{equation}
Let $\fI_N\subset\fB$ be a two-sided differential ideal generated by the 
polynomials $\fF_{2N+2}$ and $D^k(\fF_{2N+2})$,\, $k\in\mathbb{N}$.
The quotient ring $\fB_N = \fB/{\fI_N}$ is a  graded finitely generated free 
ring $\fB_N = \mathcal{A}\langle u_0,u_1,\ldots,u_{2N-1}\rangle$ over $\cA$.
Similar to Proposition \ref{prop2} we can show that equations of the hierarchy 
(\ref{kdv_hn}) are all compatible, therefore
$D_{2k-1}(\fI_N)\subset \fI_N$. The derivations $D_{2k-1}$ 
induce derivations $\cD_{ {2k-1}}$ of the quotient algebra $\fB_N$.

The equation $\fF_{2N+2} = 0$ can be written in the form
\begin{equation}\label{nn2}
 u_{2N} = \mathcal{G}_{2N+2}(u_0,u_1,\ldots,u_{2N-2}) = \widehat{\varrho}_{2N+2} 
+ 2^{2N+1}\sum_{k=0}^{N-1}\alpha_{2(N-k+1)}\varrho_{2k},
\end{equation}
which is called the non-commutative $N$-th Novikov equation.

In the non-commutative case the definition of first integrals has to be modified 
\cite{miksok_CMP}, \cite{Sok-20}.
Let ${\rm Span}_\cA[\fB_N,\fB_N]$ be a  $\cA$-linear subspace generated by the 
commutators of all elements in $\fB_N$.
We would like to emphasise that ${\rm Span}_\cA[\fB_N,\fB_N]$ is not an ideal in 
 $\fB_N$.

Let $\fB^\natural_N = \fB_N/ {\rm Span}_\cA[\fB_N,\fB_N]$ be the corresponding 
quotient linear space. It follows from the Leibniz rule
that derivations of $\fB_N$ are well defined on $\fB^\natural_N$. There is a 
short split exact sequence (compare with (\ref{sec19}):
\begin{equation}\label{sequence}
 0\ \rightarrow\  {\rm Span}_\cA[\fB_N,\fB_N]\ \rightarrow \ \fB_N\ \rightarrow 
\fB^\natural_N\ \rightarrow\ 0.
\end{equation}

All constructions and results associated with exact sequence \eqref{sec19} carry 
over to the case of exact sequence \eqref{sequence}.

\begin{defn}\label{deffirst}
A non-constant element $\cH\in\fB^\natural_N$ is called a first integral of the 
$N$-th 
Novikov equation (respectively of the non-commutative $N$-th Novikov hierarchy) 
if
$\cT(\cD \cH )= 0$ (resp. $\cT(\cD_{ {2k- 1}}\cH )= 0,\, 
k=1, 2,\ldots,N$).
\end{defn}

Thus, an element $\cH\in\fB_N$ is a representative of a first integral, if and 
only if $\cT(\cH)\ne 0$ and $\cT( \cD\cH)=0$.

It follows from (\ref{ukn-non}) and Corollary \ref{Cor-2} that
\begin{equation}\label{cH2N}
\widehat{\cH}_{2n+1,2N+1}=\widehat{\sigma}_{2n+1,2N+1}+\sum_{k=1}^{N-1}\alpha_{
2N-2k+2} \widehat{\sigma}_{2n+1,2k-1},\quad n\in \bbbN
\end{equation}
are first integrals of the non-commutative $N$-th Novikov equation. In the case 
of free algebra $\fB_N$ we get infinitely many first integrals,
since they are   algebraically independent. Also the KdV hierarchy 
\eqref{kdv_hn} reduced to $\fB_N$ is infinite,
since the derivations $\cD_{2k-1},\, k>N$, cannot be represented as liner 
combinations of $\cD_{2k-1},\ 1\leqslant k\leqslant N$,
with coefficients from a ring of constants (as it takes place in the commutative 
case (\ref{finiteD})).

{\bf Example.} The $N=1$ Novikov equation (Newton equation on the free algebra 
\\ $\fB_1 = \langle u, u_1\rangle$ with the force $3u^2+8\alpha_4$)
\[
\begin{array}{llll}
  \partial_{t_1}(u) &=\, u_{1}; \qquad  &\partial_{t_3}(u) &=\, 0; \\
  \partial_{t_{1}}(u_1) &=\, 3u^2 + 8\alpha_4;\qquad  &\partial_{t_3}(u_1) &=\, 
0;
\end{array}
\]
has an infinite hierarchy of commuting symmetries
 \begin{align}\label{uq5}
  16\partial_{t_5}(u) &=\, (u_1u^2+u^2u_1) - 2uu_1u - 16\alpha_4u_1; \\\nonumber
  16\partial_{t_{5}}(u_1) &=\, -(uu_1^2+u_1^2u) + 2u_1uu_1 - 48\alpha_4u^2 - 
128\alpha_4^2;\\&\nonumber\\
32\partial_{t_7}(u) &=\, 8 \alpha _4 u u_1+8 \alpha _4 u_1 u+  2 u_1 u^3-u_1^3+2 
u^3 u_1-u^2 u_1 u-u u_1 u^2;\label{uq7} \\
\nonumber  32\partial_{t_{7}}(u_1) &=\, 128 \alpha _4^2 u-8 \alpha _4 u_1^2+64 
\alpha _4 u^3-2 u^2 u_1^2+u u_1 u u_1-2 u u_1^2 u+u_1 u^2 u_1+\\
  &\nonumber\qquad\qquad\qquad\qquad\qquad\qquad\qquad\qquad\qquad\qquad  +u_1 u 
u_1 u-2 u_1^2 u^2+6 u^5;\\&\nonumber\\
 \label{uq9}  16\partial_{t_9}(u) &=\,3\alpha_4( u^2 u_1-2 u u_1 u+u_1 u^2-8 
\alpha _4 u_1); \\
\nonumber  16\partial_{t_{9}}(u_1) &=\,3\alpha_4( 64 \alpha _4^2+24 \alpha _4  u 
^2+ u u_1 ^2-2 u_1  u u_1+u_1^2  u).
\end{align}

There are infinitely many algebraically independent first integrals (in the 
sense of Definition \ref{deffirst}) given by (\ref{cH2N}).  For example, it 
follows from (\ref{cH2N}) that  in ${\fB}_1^\natural$
 \begin{eqnarray}
  \label{cH6}
  \widehat{\cH}_{3,3} &=& -\frac{3}{16}\left(\frac{1}{2}u_1^2- u^3-8 \alpha _4 
u\right);\\
 \cT({ \widehat{\cH}_{ 5,3}})&=&{\frac{5}{128} \cT\left(64 \alpha _4^2+5  
u u_1^2-10 u_1  u u_1+5 u_1^2  u\right)}=
  \frac{5}{2}\alpha_4^2;\\
\cT({\widehat{\cH}_{7,3}})& =& \frac{7}{265}(u_1 u u_1 u- u_1^2 
u^2)-\frac{7}{3}\alpha_4 \widehat{\cH}_{3,3};\\
        \cT({\widehat{\cH}_{9,3}})& =&        \frac{9}{512} \left( u_1 u u_1 
u^2-  u_1 ^2 u^3\right)-\frac{9}{2} \alpha _4^3+\frac{1}{2} 
\cT({\widehat{\cH}_{3,3}^2}).
 \end{eqnarray}

We have
 \begin{align*}
 & \cD (\widehat{\cH}_{3,3}) = -\frac{3}{64}( u^2 u_1-2u u_1 u+ u_1 u^2); \qquad 
\ 
\cD ( \cT({ \widehat{\cH}_{ 5,3}}))=0;
  \\
 & \cD ( \cT({ \widehat{\cH}_{7,3}})) = \frac{7}{256} \left(8 \alpha _4 u u_1 
u-8 \alpha _4 u_1 u^2+u_1 u u_1^2-u_1^2 u u_1+3 u^3 u_1 u-3 u ^2 u_1 
u^2\right)\\&-\frac{7}{3}\alpha_4 \cD (\widehat{\cH}_{3,3}),
 \end{align*}
and $\cT({\cD (\widehat{\cH}_{3,3})}) =\cT({\cD (\widehat{\cH}_{5,3})}) = 
\cT({\cD (\widehat{\cH}_{7,3})} ) = \cT({\cD (\widehat{\cH}_{9,3})}) = 0$.

\subsection{Self-adjointness of the KdV and $N$--Novikov hierarchies.} 
\label{sec2-4}  \text{}

Let $\fB_\cA = \cA\langle u,u_1,\ldots\rangle$. The set of differential 
operators
\[
\fB_\cA[D] = \left\{ A_+ = \sum_{i=0}^{m}a_iD^i\, |\,a_i\in \fB_\cA,\,a_m\neq 
0,\,m\in \mathbb{Z}_{\geqslant 0} \right\}
\]
and the set of differential formal series
\[
\fB_\cA^D = \fB_\cA[D][[D^{-1}]] = \left\{ A = \sum_{i\leqslant m}a_iD^i\, 
|\,a_i\in \fB_\cA,\,a_m\neq 0,\,m\in \mathbb{Z} \right\}
\]
are non-commutative associative algebras in which multiplication is defined by 
formula \eqref{fA-2}. According to formula \eqref{fA-2}, a conjugation 
anti-automorphism
\[
\dag \colon \fB_\cA^D \to \fB_\cA^D\,:\, (AB)^\dag = B^\dag A^\dag
\]
is defined on the ring $\fB_\cA^D$.

\begin{lem} \label{Lem-11}
{\rm 1.} The operator $\dag$ is uniquely defined by the conditions 
$$u^\dag = 
u,\, \quad D^\dag = -D,\, \quad \alpha_{2k}^\dag = \alpha_{2k},\, \quad z^\dag 
= \bar{z},\ z\in 
\mathbb{C}.$$

{\rm 2.} On the ring $\fB_\cA$ the operators $D$ and $\dag$ commute, i.e.
\begin{equation}\label{x}
  (D(a))^\dag = D(a^\dag)\; \text{ for any } a\in \fB_\cA.
\end{equation}
\end{lem}

\begin{proof}
We have $u_1 = Du-uD$. Therefore, $u_1^\dag = -uD+Du = u_1$. Using the formula 
$u_{k+1} = Du_k-u_kD$, by induction on $k,\,k\geqslant 1$, we obtain that 
$u_{k+1}^\dag = u_{k+1}$. Using now the formula $DD^k = D^{k+1}$, by induction 
on $k\geqslant 1$ and $k\leqslant -1$ we obtain that $(D^k)^\dag = (-1)^kD^k$, 
where $k\in \mathbb{Z}$ and $D^0=1$ is the identity operator.

Statement 1 now follows from the fact that elements $u_k,\, k\geqslant 0$,\, 
$D$ 
and $D^{-1}$, where $DD^{-1} = 1$, multiplicatively generate the ring 
$\fB_\cA^D$ as a module over the ring $\cA$.

Assertion 2 is verified directly. Let $a\in \fB_\cA$. Then
\[
(D(a))^\dag = (Da-aD)^\dag = -a^\dag D + Da^\dag = D(a^\dag).
\]
\end{proof}

A differential series $A\in \fB_\cA^D$ is called self-adjoint if $A^\dag = A$, 
and anti-self-adjoint if $A^\dag = -A$.

\vskip .2cm
{\bf Examples.} \text{}

{\bf 1.} The operator $L= D^2-u$ is self-adjoint.

{\bf 2.} The series $\mathcal{L} = D-\frac{1}{2}uD^{-1}+\ldots,\, \mathcal{L}^2 
= L$, is anti-self-adjoint.

Let $A,B \in \fB_\cA^D$. Then
\[
[A,B]^\dag = (AB)^\dag - (BA)^\dag = B^\dag A^\dag - A^\dag B^\dag = 
-[A^\dag,B^\dag].
\]

{\bf 3.} Let $A$ and $B$ be self-adjoint (or anti-self-adjoint) series. Then 
series $[A,B]$ is anti-self-adjoint.

{\bf 4.} Let one of the series $A$ or $B$ be self-adjoint and the other be 
anti-self-adjoint. Then series $[A,B]$ is self-adjoint.

\vskip .2cm
Consider the series $A = \sum\limits_{i\leqslant m}a_iD^i$. From the formula 
\eqref{fA-2}, we obtain
\begin{equation}\label{a+k}
A^\dag = \sum_{i\leqslant m}(-1)^i D^ia_i^\dag = \sum_{i\leqslant m}(-1)^i 
\sum_{j\geqslant 0}\binom{i}{j}D^j(a_i^\dag)D^{i-j}.
\end{equation}
According to \eqref{a+k}, from conditions $j\geqslant 0$ and $i=j-1$ we obtain 
that $j=0$ and $i=-1$, i.e.
\begin{equation}\label{resA}
  \res (A^\dag) = -(\res A)^\dag.
\end{equation}
The representation $A = A_+ + A_-$ is uniquely defined. So according to 
\eqref{a+k}, we have
\begin{equation}\label{A+A-}
  (A_+)^\dag = (A^\dag)_+\  \,\; \text{ and }\; (A_-)^\dag = (A^\dag)_-
\end{equation}

\begin{cor}\label{Cor-4}
$\varrho_{2k}^\dag = \varrho_{2k},\; k\geqslant 1$.
\end{cor}
\begin{proof}
$\varrho_{2k}^\dag = \big(\res \mathcal{L}^{2k-1}\big)^\dag = -\res\big( 
(\mathcal{L}^{2k-1})^\dag\big) = \res \mathcal{L}^{2k-1} = \varrho_{2k}$.
\end{proof}

Explicit formulas $\varrho_{2k}^\dag = \varrho_{2k}$ in the case $k=1, 2, 3, 4$ 
see \eqref{qden-1} and \eqref{qden-2}.

\begin{lem} \label{Lem-12}
The operators $\partial_{t_{2k-1}},\,k=1,2,\ldots$, participating in the 
noncommutative hierarchy KdV, commute with the conjugation operator $^\dag$, 
i.e.
\begin{equation}\label{com}
\big( \partial_{t_{2k-1}}(a) \big)^\dag = \partial_{t_{2k-1}}(a^\dag)
\end{equation}
for any $a \in \fB_\cA$.
\end{lem}
\begin{proof}
Applying the Leibniz rule for the operator $\partial_{t_{2k-1}}$, we find that 
it suffices to prove formula \eqref{com} only for multiplicative generators 
$u_i,\, i\geqslant 0$, of the ring $\fB_\cA$. Since $u_i = D^i(u)$ and the 
operators $\partial_{t_{2k-1}}$ commute with the operator $D$, it suffices to 
verify formula \eqref{com} only for the case $a=u$.  By the definition of the 
hierarchy KdV, on a free associative algebra we have $\partial_{t_{2k-1}}(u) = 
-2D(\varrho_{2k})$ where by definition $\varrho_{2k} = \res 
\mathcal{L}^{2k-1}$. 
Using the formula \eqref{x} and Corollary \ref{Cor-3}, we obtain
\[
\big( \partial_{t_{2k-1}}(u) \big)^\dag = -2\big(D(\varrho_{2k})\big)^\dag = 
-2D(\varrho_{2k}^\dag) =  -2D(\varrho_{2k}) = \partial_{t_{2k-1}}(u).
\]
\end{proof}

Let us now turn to the case of the $N$--Novikov hierarchy described in Section 
\ref{sec2-3}.

\begin{lem} \label{Lem-13}\text{}

{\rm 1.} For any $N\geqslant 1$, the conjugation operator $\dag\colon \fB_N \to 
\fB_N$ is defined.

{\rm 2.} The operators  $D_{2k-1}$, of the $N$--Novikov hierarchy, commute with 
the operator $\dag$.
\end{lem}
\begin{proof}
We have $\fB_N = \fB/\fI_N$ where $\fI_N$ is a two-sided ideal generated by 
polynomials $\fF_{2N+2}$ and $D^k(\fF_{2N+2}),\, k\in \mathbb{N}$.
Since $\varrho_{2k}^\dag = \varrho_{2k}$ and $\big(D(a)\big)^\dag = D(a^\dag)$, 
then all generators of the ideal $\fI_N$ are self-adjoint polynomials. This 
proves assertion 1.

Assertion 2 follows directly from Lemma \ref{Lem-12}.
\end{proof}

Let us consider the short exact sequence \eqref{sequence}. The operator 
$\dag\colon \fB_N \to \fB_N$ moves the linear space ${\rm 
Span}_\cA[\fB_N,\fB_N]$ into itself. Therefore, the conjugation operator 
$\dag\colon \fB^\natural_N \to \fB^\natural_N$ is defined. The first integrals 
of the $N$--Novikov hierarchy are given by formula \eqref{cH2N}, where by 
definition $\widehat{\sigma}_{2n+1,2k-1} = 
\widehat{\sigma}(\mathcal{L}_+^{2n+1},\mathcal{L}_-^{2k-1})$.

\begin{lem} \label{Lem-14}
The self-adjoint polynomials
\begin{equation}\label{Hhat}
 \widehat{\cH}_{2n+1,2N+1} = \sum_{k=1}^{N+1} \alpha_{2N-2k+2}\, 
\widehat{\sigma}_{2n+1,2k-1},\; \text{ where } \alpha_0=1,
\end{equation}
are  first integrals of the $N$--Novikov hierarchy.
\end{lem}
The proof follows from Corollary \ref{Cor-2} and Lemma 
\ref{Lem-13}.

 \section{Quantisation of Novikov's equations.}\label{sec3}

\subsection{Quantum Novikov equations.}\text{}

In this section we define  a quantisation ideal $\fQ_N$ and the quantum $N$-th 
Novikov equation. Let $\cQ_N$ be a commutative graded algebra of parameters
\[
 \cQ_N=\bbbC[\alpha_{2j+2}, q_{i,j}, q_{i,j}^\omega\,|\; 0\leqslant i<j\leqslant 
2N-1,\; 0\leqslant |\omega|<i+j+4]
\]
where $|q_{i,j}|=0$,\; $\omega=(i_{2N-1},\ldots,i_1,i_0)\in \bbbZ_\geqslant 
^{2N}$,\;
$|\omega|=(2N+1)i_{2N-1}+\cdots + 3i_1+2i_0$,\; 
$|q_{i,j}^\omega|=i+j+4-|\omega|$.
The parameters are constant in a sense that for any $a\in\cQ_N$ we have   
$\cD_{2k-1}(a)=0$.
Let
$u^\omega=u_{2N-1}^{i_{2N-1}}\cdots  u_{1}^{i_1} u_{0}^{i_0}$, and thus 
$|u^\omega|=|\omega|$. 
Let $\fB_N(q)$ denote the graded ring
\[
 \fB_N(q)=\cQ_N \langle u_0,\ldots,u_{2N-1} \rangle .
\]
with parameters, which are in the centre of the ring.

Let  $\fQ_N=\langle p_{i,j}\,|\, 0\leqslant i < j\leqslant 2N-1 
\rangle\subset\fB_N(q)$ be a two-sided $\cD=\cD_{ 1}$ differential homogeneous 
ideal
generated by the polynomials
\begin{equation}\label{dp-1}
 p_{i,j} = u_iu_j - q_{i,j}u_ju_i + \sum\limits_{0 \leqslant|\omega|< i+j+4} 
q_{i,j}^\omega u^\omega,\quad 0\leqslant i < j\leqslant 2N-1,\quad q_{i,j} \neq 
0.
\end{equation}
Let us consider the graded associative algebra $\fC_N = \fB_N(q)/\fQ_N$ and the 
ring epimorphism $\fB_N(q)\to\fC_N$ preserving the grading.

\begin{defn}
The ideal $\fQ_N$ is called the Poincar\'e--Birkhoff--Witt ideal (briefly, the 
PBW-ideal) if the image of the set of monomials
$u^\omega,\ \omega \in \bbbZ_\geqslant^{2N}$, forms a non-degenerate additive 
basis in the $\cQ$-module $\fC_N$.
\end{defn}

\begin{defn}\label{def-17}
The ideal $\fQ_N$ is a \emph{quantisation ideal} of the $N$-th Novikov equation 
if:
\begin{enumerate}
  \item[1.] it is the PBW-ideal;
  \item[2.] it is invariant with respect to the derivation $\cD$.
\end{enumerate}
\end{defn}

Condition 2 of Definition \ref{def-17} reduces to a system 
of polynomial algebraic equations in  $\cQ_N$.

\begin{lem}\label{Lem-5}
Let  $\fQ_N\subset\fB_N$ be a quantisation ideal of the $N$-th Novikov equation.
Then $q_{i,j} = 1$.
\end{lem}
\begin{proof} The Lemma can be proven by  induction.
Let us show that $q_{2N-2,2N-1}=1$. Applying $\cD$ to the polynomial 
$p_{2N-2,2N-1}$ we get
\[
\cD (p_{2N-2,2N-1}) = (1-q_{2N-2,2N-1})u_{2N-1}^2 + f_{2N-2,2N-1}.
\]
where $f_{2N-2,2N-1}$ is a polynomial whose leading monomial ${\rm 
Lm}(f_{2N-2,2N-1})<  u_{2N-1}^2$.
Thus $q_{2N-2,2N-1} = 1$ is the necessary condition for 
$\cD(\fQ_N)\subset\fQ_N$. Let us assume that $q_{i,j} = 1$
for all $i<j$, such that $i+j\geqslant k$. For a polynomial $p_{i,j}$ with $i<j$ 
and $i+j< k$ we get
\[
\cD (p_{i,j}) = (u_{i+1}u_j + u_iu_{j+1}) - q_{i,j}(u_{j+1}u_i + u_ju_{i+1}) + 
\sum q_{i,j}^\omega \cD (u^\omega).
\]
By the induction assumption  $q_{i,j+1} = 1$. Therefore, $\cD (p_{i,j}) = 
(1-q_{i,j})u_{j+1}u_i + f_{i,j}$,
where $f_{i,j}$ is a polynomial such that ${\rm Lm}(f_{i,j})\prec u_{j+1}u_i$ in 
the additive basis of  $\fC_N$.
It follows from $\cD (p_{i,j})\in \fJ_N$ that $q_{i,j} = 1$.
\end{proof}

\begin{cor}\text{}

{\bf 1.} The relation
\[
[u_i,u_j] = -h_{ij},\quad 0\leqslant i<j\leqslant 2N-1,
\]
holds in the ring $\fC_N$, where
\[
h_{ij} = \sum_{0\leqslant|\omega|<i+j+4} q_{i,j}^\omega u^\omega, \quad 
q_{i,j}^\omega\in \cQ_N.
\]

{\bf 2.} For any $P\in \fC_N$, polynomial $[u_k,P]\in \fC_N$ is a linear 
combination of polynomials $h_{ij}$ with coefficients from $\fC_N$.
\end{cor}

Since the ideal $\fQ_N\subset\fB_N(q)$ is $\cD$-invariant, i.e. 
$\cD(\fQ_N)\subset\fQ_N$, then the derivation $\cD$
induces a well defined  derivation $\partial_{t_1}$ on the  quotient algebra 
$\fB_N(q)/\fQ_N$ and a quantum dynamical system defined by the quantum $N$-th 
Novikov equation
\[
\partial_{t_1}(u_k) = u_{k+1},\; k=0,\ldots,2N-2,\qquad  
\partial_{t_1}(u_{2N-1}) = \fG_{2N+2}(u_0,\ldots,u_{2N-2}),
\]
where $\fG_{2N+2}(u_0,\ldots,u_{2N-2})\in \fB_N(q)/\fQ_N$ is a canonical 
projection of \\$\cG_{2N+2}(u_0,\ldots,u_{2N-2})\in \fB_N$.

\begin{thm}\label{T-id}
The ideal $\fQ_N$ is the quantisation ideal of the $N$-th Novikov equation if 
and only if the set of polynomials $h_{ij}\in\fC_N$
is a solution of the following systems in the ring $\fC_N$:
\begin{enumerate}
\item[I.] the system of algebraic equations linear in $h_{ij}$
\begin{equation}\label{eq-1}
[h_{ij},u_k] + [h_{jk},u_i] = [h_{ik},u_j]
\end{equation}
for all triples $(i,j,k), \; 0\leqslant i<j<k\leqslant 2N-1$;
\item[II.] the system of differential equations linear in $h_{ij}$
\begin{equation}\label{eq-2}
h_{ij}' = \widetilde h_{i+1,j} + \widetilde h_{i,j+1}
\end{equation}
where $h_{ij}' = \partial_{t_1}(h_{ij})$ and
\[
\widetilde h_{i+1,j} = \begin{cases}
                        h_{i+1,j}, & \mbox{if }\, i+1<j; \\
                         0, & \mbox{if }\, i+1=j;
                       \end{cases} \qquad
\widetilde h_{i,j+1} = \begin{cases}
                        h_{i,j+1}, & \mbox{if }\, j+1<2N; \\
                        [u_i,\fG_{2N+2}], & \mbox{if }\, j=2N-1.
                       \end{cases}
\]
\end{enumerate}
\end{thm}
\begin{proof}
In \cite{Lev-05}, V.~Levandovskyy obtained necessary and sufficient conditions 
on polynomials $p_{i,j}$ of the form \eqref{dp-1}
under which the ideal $\fQ_N$ is the BPW-ideal. Under the additional condition 
$q_{i,j}=1$ (see Lemma \ref{Lem-5}),
Lemma 2.1 from \cite{Lev-05} provides a proof of Statement I.
The proof of Statement II follows from Statements 1–2 of Corollary 4.
\end{proof}

Let the ideal $\fQ_N\subset \fB_N$ be a quantization ideal that is invariant 
under derivations $\cD_{2k-1},\, k=2,\ldots,N$,
on the ring $\fB_N$. Then derivations $\partial_{t_{2k-1}},\, k=2,\ldots,N$, are 
defined on the ring $\fC_N$ such that
$\partial_{t_{2k-1}}(u) = \partial_{t_1}(\widehat{\varrho}_{2k})$, where 
$\widehat{\varrho}_{2k}\in \fC_N$ is the image
of the polynomial $\varrho_{2k}\in \fB_N$.

\begin{cor}
The polynomials $h_{ij}\in\fC_N$ satisfy the following system of differential 
equations linear in $h_{ij}$
\[
\partial_{t_{2k-1}}(h_{ij}) = [u_i,\partial_{t_1}^{j+1}(\widehat{\varrho}_{2k})] 
- [u_j,\partial_{t_1}^{i+1}(\widehat{\varrho}_{2k})],\quad
k=1,\ldots,N,\; 0\leqslant i<j\leqslant 2N-1.
\]
\end{cor}

\subsection{Quantum $N=1$ Novikov equation.}\label{sec32}\text{}

Let
\begin{equation}\label{f-N1}
 u_2 = \cG_4(u_0) = 3u^2 + 8\alpha_4
\end{equation}
and $\fI_1\subset \fB$ is the two-sided differential ideal generated by Novikov 
equation relatively the derivation $D$, such that $D(u_k) = u_{k+1}$.
Then in $\fB_1 = \fB/\fI_1 = \mathcal{A}\langle u_0,u_1\rangle$ the induced 
derivation $\cD$ can be defined by its action on the generators
 \[
  \cD(u) = u_1,\quad \cD(u_1) = 3u^2 + 8\alpha_4.
 \]
Let us consider a homogeneous two sided ideal $\fJ_1 \subset \fB_1(q)$ generated 
by one polynomial
\[ p_{0,1} = uu_1 - u_1u -  q^{(0,2)} u^2 -  q^{(1,0)} u_1 -  q^{(0,1)} u -  
q^{(0,0)}
\]
with arbitrary constants $q^\omega$. Let us find conditions on these constants 
under which the ideal $\fQ_1$ becomes the quantization ideal.
According to Theorem \ref{T-id} we obtain
\[
h_{01}' = [u,3u^2+8\alpha_4] = 0.
\]
Thus, $ q^{(0,2)} =  q^{(1,0)} =  q^{(0,1)} = 0$ and $ q^{(0,0)}$ is a free 
parameter.
Let us denote $ q^{(0,0)} = 8i\hbar$, then the quotient algebra 
$\fC_1=\fB_1(q)/\fQ_1$ coincides
with the Heisenberg (Weyl) algebra $\bbbC[\alpha_4,\hbar]\langle u,u_1\rangle 
/\langle uu_1-u_1u-8i\hbar\rangle$
in quantum mechanics. Thus we have proved the following statement.
\begin{prop}\label{prop-21}
The ideal $\fQ_1\subset\fB_1(q)$ is the quantization ideal if and only if
$uu_1-u_1u=8i\hbar$, where $\hbar$ is an arbitrary parameter.
\end{prop}

In the case $N=1$ the $N$-th Novikov equation $u_2 = 3u^2 + 8\alpha_4$ has the 
form of the classical Newton equation.
According to Proposition \ref{prop-21}, the quantum $N=1$ Novikov equation is 
unique and can be written in the Heisenberg form
\begin{equation}\label{N1H}
 i\hbar \partial_{t_1}(u) =[u,\fh_{3,3} ]= i\hbar u_1 ,\qquad
 i\hbar \partial_{t_1}(u_1)= [u,\fh_{3,3} ]= i\hbar( 3u^2 + 8\alpha_4),
\end{equation}
where the Hamiltonian operator $\fh_{3,3}=-\frac{2}{3}\widehat{\cH}_{3,3} =
\frac{1}{16}( u_1^2 - 2u^3 - 16\alpha_4u)$ is self-adjoint and $[u,u_1] =
8i\hbar$. Here $\widehat{\cH}_{3,3}$ is given by (\ref{Hhat}). It follows from
(\ref{N1H}) that the $t_1$ derivative of any element of $a\in\fC_1$ can be
written in the form $\partial_{t_1}(u) =\frac{i}{\hbar}[\fh_{3,3},a ]$. In
particular the Hamiltonian $\fh_{3,3}$ is a quantum constant of motion
$\partial_{t_1}(\fh_{3,3}) =\frac{i}{\hbar}[\fh_{3,3},\fh_{3,3} ]=0$.

We have $\partial_{t_3}(u) = 0$. Higher symmetries (\ref{uq5}), 
(\ref{uq7}),(\ref{uq9}) on the free associative algebra $\fB_1$ after the 
reduction to $\fC_1$ take the self-adjoint form
\begin{align*}
   \partial_{t_5}(u) &=\, -  \alpha _4 u_1; \\
   \partial_{t_{5}}(u_1) &=\, -  \alpha _4 (3u^2+8\alpha_4);\\[6pt]
32\partial_{t_7}(  u  ) &=\, 64 i \alpha _4 \hbar +16 \alpha _4  u _1   u  +24 i 
  u  ^2 \hbar +2  u _1   u  ^3- u _1^3; \\
32\partial_{t_7}(  u  _1) &=\, 64 \alpha _4   u  ^3+128 \alpha _4^2   u  -8 
\alpha _4  u _1^2-48 i  u _1   u   \hbar +6   u  ^5-3  u _1^2   u  ^2+192 \hbar 
^2;\\[6pt]
 2\partial_{t_9}(u) &=\, 3\alpha_4^2 u_1; \\
 2\partial_{t_{9}}(u_1) &=\, 3\alpha_4^2(3u^2+8\alpha_4).
\end{align*}

\subsection{Quantum $N=2$  Novikov hierarchy.}\text{}

In the case $N=2$ the $N$-th Novikov  equation on the free associative 
algebra\\ $\fB =\bbbC[\alpha_4,\alpha_6]\langle u,u_1,\ldots\rangle$ can be 
written in the form
\begin{equation}\label{f-u4}
 u_{4} = \mathcal{G}_6(u_0,u_1,u_2),
\end{equation}
where
\begin{equation}\label{G-6}
  \mathcal{G}_6 = 5(u_2u + uu_2) + 5u_1^2 - 10u^3 - 16\alpha_4u + 32\alpha_6.
\end{equation}
It defines two commuting derivations $\cD, \cD_3$ on the quotient ring\\
$\fB_2=\bbbC[\alpha_4,\alpha_6]\langle u,u_1,u_2,u_3\rangle=\fB/ \fI_2$
which results in the $2$-KdV hierarchy consisting of two compatible nonabelian 
systems.
The first system of the hierarchy \big($\partial_{t_1}(u_k) = \cD(u_k)$\big)
\begin{equation}\label{Sist-1}
 \partial_{t_1}(u) = u_1;\qquad \partial_{t_1}(u_1) = u_2;\qquad 
\partial_{t_1}(u_2) = u_3;\qquad \partial_{t_1}(u_3) = \mathcal{G}_6
\end{equation}
is equation (\ref{f-u4}) written in the form of a first order system. The second 
system of the hierarchy is
\begin{equation} \label{kdv}
 4\partial_{t_3}(u_{k}) = \partial_{t_1}^{k+1}(u_2 - 3u^2),\quad k=0,1,2,3.
\end{equation}

Let us introduce a two-sided ideal
\[
\widehat{\fQ}_2 = \langle[u_i,u_j]-h_{ij};\, [h_{ij},u_k];\,0\leqslant 
i<j\leqslant 3,\, 0\leqslant k\leqslant 3\rangle\subset \fB_2(q)
\]
and set
\[
\widehat{\fC}_2 = \fB_2(q)/\widehat{\fQ}_2.
\]

\begin{lem}\label{Lem-22}
The ideal $\widehat{\fQ}_2$ is a quantisation ideal of the $2$-nd Novikov 
equation if and only if polynomials $h_{ij}\in \widehat{\fC}_2$
are the solution of the following linear in $h_{ij}$ system of differential 
equations:
\begin{align*}
h_{01}' = h_{02};\qquad\quad\,  h_{02}' = h_{12}+h_{03};\quad  h_{12}' = h_{13}; 
\\
h_{03}' = h_{13}+P_8;\quad  h_{13}' = h_{23}+P_9;\quad\; h_{23}' = P_{10},
\end{align*}
where $h_{ij}' = \partial_{t_1}(h_{ij})$, $P_8 = (h_{01}u)' = h_{02}u + 
h_{01}u_1$, $P_9 = 10(h_{12}u - h_{01}u_2)+16\alpha_4 h_{01}$,
$P_{10} = 10h_{02}(-u_2 + 3u^2) - 10h_{12}u_1 + 16\alpha_4 h_{02}$.
\end{lem}
\begin{proof}
Using the formula
\[
\mathcal{G}_6 = 10(u_2u-u^3) + 5(h_{02}+u_1^2) - 16\alpha_4u + 32\alpha_6 \in 
\widehat{\fC}_2,
\]
we obtain that the assertion of this lemma is an immediate consequence of 
Theorem \ref{T-id}.
\end{proof}

\begin{prop}\label{prop-23}
The ideal $\fQ_2\subset\fB_2(q)$ is the quantization ideal if and only if
\[
[u_i,u_j] = 0\, \text{ for } i+j<3\, \text{ or } i+j=4; \quad [u,u_3] = 
[u_2,u_1] = 32i\hbar; \quad [u_2,u_3] = 320 i\hbar u,
\]
where $\hbar$ is an arbitrary parameter.
\end{prop}
\begin{proof}
Set $h_{01} = \xi\in \widehat{\fC}_2$. Then, accordingly to Lemma \ref{Lem-22}, 
in the ring $\widehat{\fC}_2$ we have:
\begin{align*}
h_{02} &= \xi';\quad 2h_{03} = \xi''+10\xi u+2\alpha,\,\text{ where }\, 
\alpha=const,\, |\alpha|=7; \\
2h_{12} &= \xi''-10\xi u-2\alpha;\quad 2h_{13} = \xi'''-10(\xi u)'; \\
2h_{23} &= \xi^{(4)}-10(\xi u)''-20(h_{12}u-\xi u_2)-16\alpha_4\xi; \\
\xi^{(5)} &- 10(\xi u)'''-20(h_{12}u-\xi u_2)'-16\alpha_4\xi' = 20\xi'(-u_2 + 
3u^2)-20h_{12}u_1 + 32\alpha_4\xi'.
\end{align*}
Therefore, the polynomial $\xi$ must satisfy the equation
\begin{equation}\label{eq-76}
\xi^{(5)}-20\xi'''u-30\xi''u_1+\xi'(10u_2+40u^2-48\alpha_4)+10\xi(u_3+10u_1u) = 
0.
\end{equation}
Accordingly to formula \eqref{dp-1}, the solution to this equation should be 
sought in the form:
\begin{equation}\label{eq-77}
\xi = \beta_1u^2+\beta_2u_1+\beta_3u+\beta_5,\quad |\beta_k| = k.
\end{equation}
Substituting expression \eqref{eq-77} into equation \eqref{eq-76} and using the 
PBW-basis in $\widehat{\fC}_2$,
we obtain that $\xi=0$ in $\widehat{\fC}_2$.
Then the polynomial $\xi\in\fB_2(q)$ must belong to the ideal 
$\widehat{\fQ}_2$, 
but this is possible only when $[u,u_1] = 0$.

Under condition $[u,u_1] = 0$, the system of differential equations \eqref{eq-2} 
in the case $N=2$ takes the form
\begin{align*}
h_{01} &= 0;\quad h_{02} = 0;\quad h_{12}+h_{03} = 0; \\
h_{13} &= h_{12}' = h_{03}';\quad h_{23} = h_{13}'-5(h_{12}u+uh_{12}); \\
h_{23}' &= -5(h_{12}u_1+u_1h_{12}).
\end{align*}
If $\alpha = -32i\hbar$, then we obtain:
\[
[u,u_3] = [u_2,u_1] = 32i\hbar; \quad [u_2,u_3] = 320 i\hbar u.
\]
\end{proof}

\begin{prop}\label{q2}
The following statements are equivalent \\[3pt]
{\rm 1.} The ideal $\fJ_2$ is $\cD$-invariant:\; 
$\cD(\fJ_2)\subseteq\fJ_2$.\\[3pt]
{\rm 2.} The ideal $\fJ_2$ is $\cD_3$-invariant:\; 
$\cD_3(\fJ_2)\subseteq\fJ_2$.\\[3pt]
{\rm 3.} The ideal $\fJ_2$ is generated by the commutation relations
\[\begin{array}{l}\phantom{.}
 [u_i,u_j]=0\ \mbox{\rm for }\ i+j<3\ \mbox{\rm or }\ i+j=4; \\\phantom{.}
 [u,u_3]=[u_2,u_1]=32i\hbar,\quad [u_2,u_3]=320i\hbar u_0,
  \end{array}
\]
where $\hbar$ is an arbitrary parameter.
\end{prop}

\begin{prop}\label{hams2}
Quantum $N=2$  KdV hierarchy has the quantum Hamiltonians
\[
 \fh_{3,5} =-\frac{2}{3}\widehat{\cH}_{3,5},\qquad \fh_{5,5} 
=-\frac{2}{5}\widehat{\cH}_{5,5}
\]
such that $[ \fh_{3,5}, \fh_{5,5}]=0$. Here 
$\widehat{\cH}_{3,5},\,\widehat{\cH}_{5,5}$
are given by (\ref{Hhat}).
\end{prop}

\begin{thm}\label{T-1}
For $N=2$ the quantum $N$-th Novikov equation, corresponding to the derivations
$\partial_{t_1} $, can be written in the Heisenberg form
\[
\partial_{t_1}(u_k)  = \frac{i}{\hbar}[ \fh_{3,5},u_k] =
\begin{cases}
u_{k+1},\quad 0\leqslant k\leqslant 2, \\
32\alpha_6 - 16\alpha_4u  + 5u_{1}^2 + 10u_{2}u  - 10u ^3, \; k=3.
\end{cases}
\]
The quantum dynamical system  $\fC_2$, corresponding to the derivations
$\partial_{t_3}$, can be written in the Heisenberg  form
\[
 4\partial_{t_3}(u_k)  = \frac{i}{\hbar}[\fh_{5,5},u_k ] =
\partial_{t_1}^{k}(u_3 - 6u_1u),\qquad k=0,1,2,3 .
\]
\end{thm}

\subsection{Quantum $N=3$ and $N=4$ Novikov hierarchy.}\text{}

For $N=3$ non-commutative $N$-th Novikov's equation has a form
\begin{equation}\label{f-u6}
 u_{6} = \mathcal{G}_8(u_0,u_1,u_2,u_3,u_4),
\end{equation}
where
\begin{multline}\label{G-8}
  \mathcal{G}_8 = 7(u_4u+uu_4) + 14(u_3u_1+u_1u_3) + 21u_2^2 - 21(u_2u^2+u^2u_2) 
- 28(u_1^2u+uu_1^2) - \\
 - 28uu_2u - 14u_1uu_1 + 35u^4 - 16\alpha_4(u_2-3u^2) - 64\alpha_6u + 
128\alpha_8 = \mathcal{G}_8^\dag.
\end{multline}
In $\fB_3 = \mathcal{A}\langle u_0,u_1,\ldots,u_5\rangle$ the non-commutative 
$N$-th Novikov equation is a generator of the two-sided ideal $\fI_3$.

\begin{prop}\label{q3}
The following statements are equivalent \\[3pt]
{\rm 1.} The ideal $\fQ_3$ is $\cD$-invariant:\; 
$\cD(\fQ_3)\subseteq\fQ_3$.\\[3pt]
{\rm 2.} The ideal $\fQ_3$ is $\cD_3$-invariant:\; 
$\cD_3(\fQ_3)\subseteq\fQ_3$.\\[3pt]
{\rm 3.} The ideal $\fQ_3$ is $\cD_5$-invariant:\; 
$\cD_5(\fQ_3)\subseteq\fQ_3$.\\[3pt]
{\rm 4.} The ideal $\fQ_3$ is generated by the commutation relations
\[
 \begin{array}{l}\phantom{.}
[u_i,u_j]=0\ {\rm if}\ i+j<5\ {\rm or}\ i+j=6,\\\phantom{.}
[u,u_5]=[u_4,u_1]=[u_2,u_3]=\eta,\quad [u_2,u_5]=[u_4,u_3]=7\cdot
2\eta u,\\\phantom{.}
[u_5,u_3]=14\eta u_1,\quad [u_4,u_5]= 2\eta \hbar(63 u^2+14 u_2-8\alpha_4)
 \end{array}
\]
where $\eta\in\bbbC$ is an arbitrary parameter.
\end{prop}
Setting  $\eta =2^7 i\hbar$, where $\hbar$ is an arbitrary real parameter, we 
get. 
\begin{prop}
Quantum $N=3$  KdV hierarchy has three self-adjoint commuting quantum 
Hamiltonians
\[
 \fh_{3,7} = -\frac{2}{3}\widehat{\cH}_{3,7}\,,\quad \fh_{5,7}
=-\frac{2}{5}\widehat{\cH}_{5,7},\quad
 \fh_{7,7} = -\frac{2}{7}\widehat{\cH}_{7,7}
\]
\end{prop}

\begin{thm}\label{T-2}
For $N=3$ the quantum $N$-th Novikov equation, corresponding to the derivation
$\partial_{t_1}$, can be written in the  Heisenberg form
\[
\partial_{t_1}(u_k)  = \frac{i}{\hbar}[ \fh_{3,7},u_k] =
\begin{cases}
u_{k+1}, & 0\leqslant k\leqslant 4, \\
\mathcal{\fG}_8, & k=5.
\end{cases}
\]
where
\[
 \fG_8=  28 u_3 u_1 + 21 u_2^2 + 35 u^4-14(u_4\! -\! 5u_2 u\! +\! 5u_1^2)u  - 16
\alpha _4(u_2\! -\! 3u^2) - 64 \alpha_6 u + 128 \alpha_8.
\]

The quantum dynamical systems in $\fC_3$, corresponding to the derivations
$\partial_{t_3}$ and $\partial_{t_5}$ can be written in the Heisenberg  form
\begin{align*}
4\partial_{t_3}(u_k)  &= \frac{i}{\hbar}[\fh_{5,7},u_k] =
\partial_{t_1}^{k}(u_3-6u_1u), \\
16\partial_{t_5}(u_k) &= \frac{i}{\hbar}[\fh_{7,7},u_k] =
\partial_{t_1}^{k}(u_{5} - 20u_{2}u_{1} - 10u_{3}u  + 30u_{1}u ^2).
\end{align*}
\end{thm}

In the case ${\bf N=4}$ the invariant ideal of quantisation $\fQ_4$ is generated 
by the commutation relations ($\eta=2^9 i\hbar$): 
\[
\begin{array}{l}
[u ,u_7] = [u_6, u_1] = [u_2 ,u_5] = [u_4 ,u_3] = \eta ,\quad [u_2, u_7] = 
[u_6,u_3] = [u_4, u_5] = 18 \eta  u,\\[5pt]
[u_7, u_3]=2[u_4,u_6] = 36 \eta  u_1,\quad [u_4,u_7]-18\eta  u_2 = [u_6,u_5] = 
\eta (198u^2+60u_2-16\alpha_4),\\[5pt]
[u_7 ,u_6] = \eta (858u_1^2-1980u_2 u-1716u^3-54u_4+64\alpha_6+416\alpha_4 
u),\\[5pt]
[u_7 ,u_5] = \eta (396u_1 u+60u_3),\quad [u_i,u_j] = 0\; {\rm if}\; i+j<7\; {\rm 
or}\; i+j=8.
  \end{array}
\]
The corresponding quantum hierarchy takes the form 
\[
 2^{2n-2}\partial_{t_{2n-1}}(u_k)   = \frac{i}{\hbar}[\fh_{2n+1,9},u_k],\qquad 
n=1,2,3,4,\ \ k=0,\ldots ,7\, ,
\]
where
\[
 \fh_{2n+1,9} = -\frac{2}{2n+1}\widehat{\cH}_{2n+1,9},\quad n=1,2,3,4.
\]
The first equation of this fine hierarchy is the $N=4$ Novikov equation written 
in the Heiseberg form
\[
\partial_{t_1}(u_k)  = \frac{i}{\hbar}[ \fh_{3,9},u_k] =
\begin{cases}
u_{k+1}, & 0\leqslant k\leqslant 6, \\
\mathcal{\fG}_{10}, & k=7.
\end{cases}
\]
where
\begin{eqnarray*}
 &&\fG_{10}=512 \alpha _{10}-160 \alpha _4   u  ^3+192 \alpha _6   u  ^2+160 
\alpha _4  u _2   u  -256 \alpha _8   u  +80 \alpha _4  u _1^2\\&&-16 \alpha _4 
 u _4-64 \alpha _6  u _2-126   u  ^5+420  u _2   u  ^3+630  u _1^2   u  
^2-126  u _4   u  ^2\\&&-378  u _2^2   u  -504  u _1  u _3   u  +18  u _6   u  
+69  u _3^2-462  u _1^2  u _2+114  u _2  u _4+54  u _1  u _5 ,
\end{eqnarray*}
and the other three equations are ($k=0,\ldots,7$):
\begin{align*}
4\partial_{t_3}(u_k)  &= \frac{i}{\hbar}[\fh_{5,9},u_k] =
\partial_{t_1}^{k}(u_3-6u_1u), \\
16\partial_{t_5}(u_k) &= \frac{i}{\hbar}[\fh_{7,9},u_k] =
\partial_{t_1}^{k}(u_{5} - 20u_{2}u_{1} - 10u_{3}u  + 30u_{1}u ^2),\\
64\partial_{t_7}(u_k) &= \frac{i}{\hbar}[\fh_{9,9},u_k] =\partial_{t_1}^{k}(
 u _7-140  u _1   u  ^3+70  u _3   u  ^2+280  u _2  u _1   u  -14  u _5   u \\& 
+70  u _1^3-70  u _3  u _2-42  u _4  u _1)
.
\end{align*}

There are two 
interesing observations. In all cases considered in this section, namely 
$N=1,2,3,4$:
\begin{itemize}
 \item  The form of the normally ordered quantum N-Novikov equations coincide 
with the corresponding  classical equations in 
the commutative case.  
\item For the first $N$ equations of the quantum $N$-th Novikov hierarchy the 
polynomials
$\widehat{\cH}_{2n+1,2N+1}, \ n=1,2,\ldots ,N$ (\ref{Hhat}) are also quantum 
commuting Hamiltonians. We have shown that these polynomials are integrals for 
the non-commutative hierarchy on the free associative algebra $\fB_N$, i.e. 
$\partial_{t_{2k-1}}(\widehat{\cH}_{2n+1,2N+1})\in\Span[\fB_N,\fB_N]$, where 
$\partial_{t_{2k-1}},\ k=1,\ldots,N$ . Apparently these derivations   map these 
polynomials (naturally embedded in $\fC_N$) into the corresponding quantisation 
ideal $\fQ_N$. In other words 
\[ \partial_{t_{2k-1}}:\widehat{\cH}_{2n+1,2N+1}\mapsto 
\Span[\fB_N,\fB_N]\bigcap \fQ_N,\qquad 1\leqslant n,k\leqslant N.\]
\end{itemize}

\subsection*{Acknowledgements}

We would like to thank S.P. Novikov and A.P. Veselov for useful discussions.
A.V.~Mikhailov is grateful to the EPSRC (Small Grant Scheme EP/V050451/1).

\label{lastpage}
\end{document}